\documentclass[11pt]{article}
\usepackage{indentfirst}
\usepackage{amsfonts}
\usepackage{amsmath}
\usepackage{amsthm}
\usepackage{color}
\usepackage[backref=page]{hyperref}
\usepackage[capitalize]{cleveref}
\usepackage{verbatim}
\usepackage[makeroom]{cancel}
\usepackage{rotating}
\usepackage{float}
\usepackage{fullpage}
\usepackage{amssymb}
\usepackage{amsthm}
\usepackage{stmaryrd}
\usepackage{mathrsfs}
\usepackage{epsfig}
\usepackage{indentfirst}
\usepackage{framed}
\usepackage[numbers]{natbib}
\usepackage{color}
\usepackage{graphicx}
\usepackage{float}
\usepackage{bbm}
\usepackage{mathtools}
\usepackage{enumitem}
\usepackage{thmtools}
\usepackage{thm-restate}
\usepackage{longtable}
\usepackage{booktabs}
\usepackage{./shortcuts}
\usepackage{./shortcuts_alphabet}
\usepackage{setspace}
\usepackage[explicit]{titlesec}

\titleformat{\subsubsection}[runin]{\normalfont\bfseries}{\thesubsubsection.}{1em}{#1.}
\titleformat{\paragraph}[runin]{\normalfont\bfseries}{.}{1em}{#1}
\titleformat{\subparagraph}[runin]{\normalfont\itshape}{.}{0pt}{#1}
\titlespacing*{\subparagraph}{-1pt}{3.25ex plus 1ex minus .2ex}{1em}

\usepackage[labelfont=bf]{caption}

\definecolor{corlinks}{RGB}{200,0,0}
\definecolor{cormenu}{RGB}{200,0,0}
\definecolor{corurl}{RGB}{200,0,0}

\hypersetup{
colorlinks=true,
urlcolor=corlinks,
linkcolor=corlinks,
menucolor=cormenu,
citecolor=corlinks,
pdfborder= 0 0 0
}

\newcommand{\headsubsec}[1]{\paragraph{#1.}}

\RequirePackage{postlattice}
\postset{scale=-1,labels=wikipedia}
\newcommand\postlegend[1]%
 {\raisebox{-0.45ex}{\tikz\node[circle,draw,minimum size=2ex,#1]{};}}
\tikzset{NA/.style = {fill = white}}
\tikzset{pL/.style = {fill = green!70}}
\tikzset{NL/.style = {fill = yellow!70}}
\tikzset{L/.style = {fill = orange!80}}
\tikzset{LN/.style = {fill = red!70}}
\tikzset{NP/.style = {fill = cyan!80}}
\tikzset{P/.style = {fill = magenta!70}}
\definecolor{mikadoyellow}{rgb}{1.0, 0.77, 0.05}
\definecolor{cyan(process)}{rgb}{0.0, 0.72, 0.92}
\definecolor{applegreen}{rgb}{0.55, 0.71, 0.0}

\begin{document}

\title{Constant-Depth Circuits vs.~Monotone Circuits\vspace{0.5cm}}

\author{
		Bruno P. Cavalar\footnote{Email: \texttt{bruno.pasqualotto-cavalar@warwick.ac.uk}}\vspace{0.2cm}\\{\small Department of Computer Science}\\{\small University of Warwick\vspace{0.3cm}}
		\and	
		Igor C. Oliveira\footnote{Email: \texttt{igor.oliveira@warwick.ac.uk}}\vspace{0.2cm}\\{\small Department of Computer Science}\\
		{\small University of Warwick} 
	}

\maketitle

\vspace{-0.6cm}

\begin{abstract}
We establish new separations between the power of monotone and general (non-monotone) Boolean circuits:
\begin{itemize}
    \item[--] For every $k \geq 1$, there is a monotone  function in
        $\mathsf{AC^0}$ (constant-depth poly-size circuits) that requires
        monotone circuits of depth $\Omega(\log^k n)$. This significantly  extends
        a classical result of Okol'nishnikova~\cite{okolnishnikova_1982}
        and Ajtai and Gurevich \citep{ajtai_gurevich_1987}. In addition, our separation holds for a monotone graph property, which was unknown even in the context of $\AC^0$ versus $\mAC^0$.  
    \item[--] For every $k \geq 1$, there is a monotone  function in
        $\mathsf{AC^0}[\oplus]$ (constant-depth poly-size circuits extended
        with parity gates) that requires monotone circuits of size
        $\exp(\Omega(\log^k n))$. This makes progress towards a question
        posed by Grigni and Sipser \citep{gs92}.  
\end{itemize}
These results show that constant-depth circuits can be more
efficient than monotone formulas and monotone circuits when computing monotone functions.

In the opposite direction, we observe that non-trivial simulations are possible in the absence of parity gates: every monotone function computed by an $\mathsf{AC}^0$ circuit of size $s$ and depth $d$ can be computed by a monotone circuit of size $2^{n - n/O(\log s)^{d-1}}$. We 
show that the
existence of significantly faster monotone simulations would lead to
breakthrough circuit lower bounds. In particular, if every monotone
function in $\mathsf{AC}^0$ admits a polynomial size monotone circuit, then $\NC^2$
is not contained in $\NC^1$.

Finally, we revisit our separation result against monotone circuit size and
investigate the limits of our approach, which is based on a monotone lower
bound for constraint satisfaction problems (CSPs)
established by 
G\"{o}\"{o}s, Kamath, Robere and Sokolov~\citep{gkrs19}
via lifting techniques. Adapting results of Schaefer
\citep{DBLP:conf/stoc/Schaefer78} and 
Allender, Bauland, Immerman, Schnoor and Vollmer~\citep{DBLP:journals/jcss/AllenderBISV09},
we obtain an unconditional classification
of the monotone circuit complexity of Boolean-valued CSPs via their
polymorphisms. This result and the consequences we derive from it might be
of independent interest. 
\end{abstract}

\newpage  
\setcounter{tocdepth}{3}  
  
\tableofcontents

\newpage

\vspace{-0.4cm}

\section{Introduction}

A Boolean function $f \colon \{0,1\}^n \to \{0,1\}$ is monotone if $f(x) \leq f(y)$ whenever $x_i \leq y_i$ for each coordinate $1 \leq i \leq n$. Monotone Boolean functions, and the monotone Boolean circuits\footnote{Recall that in a monotone Boolean circuit the gate set is limited to
$\set{\mathsf{AND}, \mathsf{OR}}$ and input gates are labelled by elements from $\{x_1, \ldots, x_n, 0, 1\}$.} that compute them, have been extensively investigated for decades due to their relevance in circuit complexity  \citep{razborov_boolean_85}, cryptography \citep{DBLP:conf/crypto/Leichter88}, learning theory \citep{DBLP:journals/jacm/BshoutyT96},  proof complexity \citep{Krajicek97, DBLP:journals/jsyml/Pudlak97}, property testing \citep{DBLP:conf/focs/GoldreichGLR98}, pseudorandomness \citep{DBLP:conf/stoc/ChattopadhyayZ16}, optimisation \citep{gjw18}, hazard-free computations \citep{DBLP:journals/jacm/IkenmeyerKLLMS19}, and meta-complexity \citep{Hir22}, among other topics. In addition, over the last few years a number of results have further  highlighted the importance of monotone complexity as a central topic in the study of propositional proofs, total search problems, communication protocols, and related areas (see \citep{DBLP:journals/sigact/RezendeGR22} for a recent survey). 

Some of the most fundamental results about monotone functions deal with their complexities with respect to different classes of Boolean circuits, such as the monotone circuit lower bound of Razborov \citep{razborov1985lower} for $\mathsf{Matching}$ and the constant-depth circuit lower bound of Rossman \citep{DBLP:conf/stoc/Rossman08} for $k$-$\mathsf{Clique}$. Particularly important to our discussion is a related strand of research that contrasts the computational power of monotone circuits relative to general (non-monotone) $\mathsf{AND}$/$\mathsf{OR}$/$\mathsf{NOT}$ circuits, which we review next.

\headsubsec{Weakness of Monotone Circuits} The study of monotone
simulations of non-monotone computations and associated separation results
has a long and rich history.  In a sequence of celebrated results,
\citep{razborov1985lower, andreev1985method,
DBLP:journals/combinatorica/AlonB87, tardos_1988} showed the existence of
monotone functions that can be computed by circuits of polynomial size but
require monotone circuits of size $2^{n^{\Omega(1)}}$. In other words, the
use of negations can significantly speedup the computation of monotone
functions.  More recently, 
G\"{o}\"{o}s, Kamath, Robere and Sokolov~\citep{gkrs19} 
considerably
strengthened this separation by showing that some monotone functions in
$\mathsf{NC}^2$ (poly-size $O(\log^2 n)$-depth fan-in two circuits) require
monotone circuits of size $2^{n^{\Omega(1)}}$. (An earlier weaker
separation against monotone depth $n^{\Omega(1)}$ was established in
\citep{RW92}.) Therefore, negations can also allow monotone functions to be
efficiently computed in parallel. 

Similar separations about the limitations of monotone circuits
are also known at the low-complexity end of the spectrum:   Okol'nishnikova~\cite{okolnishnikova_1982} and (independently) Ajtai and
Gurevich \citep{ajtai_gurevich_1987} exhibited monotone functions in $\mathsf{AC}^0$
(i.e., constant-depth poly-size
$\mathsf{AND}$/$\mathsf{OR}$/$\mathsf{NOT}$ circuits) that require monotone $\mathsf{AC}^0$
circuits (composed of only $\mathsf{AND}$/$\mathsf{OR}$ gates)
of super-polynomial size.\footnote{We refer to \cite{BST13} for an alternate exposition of this result.} This result has been extended to an exponential separation in \citep{cos15}, which shows the existence of a monotone function in $\mathsf{AC}^0$ that requires monotone depth-$d$ circuits of size $2^{\widetilde{\Omega}(n^{1/d})}$ even if $\mathsf{MAJ}$ (majority) gates are allowed in addition to $\mathsf{AND}$/$\mathsf{OR}$ gates.\footnote{Separations between monotone and non-monotone devices have also been extensively investigated in other settings. This includes average-case complexity \citep{DBLP:conf/icalp/BlaisHST14}, different computational models, such as span programs 
\citep{DBLP:journals/combinatorica/BabaiGW99, rprc16} and algebraic complexity (see \cite{cdm21} and references therein), and separations in first-order logic \citep{DBLP:conf/lics/Stolboushkin95, DBLP:conf/lics/Kuperberg21,  DBLP:journals/corr/abs-2201-11619}. We restrict our attention to worst-case separations for Boolean circuits in this paper.}

\headsubsec{Strength of Monotone Circuits} 

In contrast to these results, in many settings negations do not offer a significant speedup and monotone computations can be unexpectedly powerful. For instance, monotone circuits are able to efficiently implement several non-trivial algorithms, such as solving constraint satisfaction problems using treewidth bounds (see, e.g.,~\citep[Chapter 3]{Oliveira15}). As another example, in the context of cryptography, it has been proved that if one-way functions exist, then there are monotone one-way functions \citep{DBLP:journals/toc/GoldreichI12}. Below we describe results that are more closely related to the separations investigated in our paper.

 In the extremely constrained setting of depth-$2$ circuits, Quine \cite{quine_1953} showed that monotone functions computed by size-$s$ $\mathsf{DNFs}$ (resp.,~$\mathsf{CNFs}$) can always be computed by size-$s$ monotone $\mathsf{DNFs}$ (resp.,~$\mathsf{CNFs}$). Some results along this line are  known for larger circuit depth, but with respect to more structured classes of monotone Boolean functions. Rossman~\cite{rossman_2008, rossman_16} showed that any homomorphism-preserving graph property computed by \(\AC^0\) circuits is also computed
by monotone \(\AC^0\) circuits.\footnote{A function $f : \blt^{\binom{n}{2}} \to \blt$
is called a \emph{graph property}
if
$f(G) = f(H)$ whenever $G$ and $H$ are isomorphic graphs,
and \emph{homomorphism-preserving}
if
$f(G) \leq f(H)$ whenever there is a 
graph homomorphism from $G$ to $H$. It is easy to see that every homomorphism-preserving graph property is monotone.} Under no circuit depth restriction, 
Berkowitz \citep{berkowitz} proved that the monotone and non-monotone circuit size complexities of every slice function are polynomially related.\footnote{A function $f : \blt^{\binom{n}{2}} \to \blt$
is a \emph{slice function} if there is $i \geq 0$ such that $f(x)$ is $0$ on inputs of Hamming weight less than $i$ and $1$ on inputs of Hamming weight larger than $i$.}\\

Despite much progress and sustained efforts, these two classes of results
 leave open tantalising problems about the power of cancellations in computation.\footnote{Any non-monotone circuit can be written as an $\mathsf{XOR}$ (parity) of distinct monotone sub-circuits (see, e.g., \citep[Appendix A.1]{DBLP:conf/tcc/GuoMOR15}), so negations can be seen as a way of combining, or cancelling, different monotone computations. See also a  related discussion in Valiant \citep{DBLP:journals/tcs/Valiant80}.} In particular, they suggest the following basic question about the contrast between the weakness of monotone computations and the strength of negations:

\begin{center}
    \emph{What is the largest computational gap between the  power of monotone and\\ general \emph{(}non-monotone\emph{)} Boolean circuits?}
\end{center}

A concrete formalisation of this question dates back to the seminal work on monotone complexity of Grigni and Sipser \citep{gs92} in the early nineties. They asked if there are
monotone functions in $\AC^0$ that require super-polynomial size monotone
Boolean circuits, i.e., if $\semMon{\mathsf{AC}^0} \nsubseteq
\mathsf{mSIZE}[\mathsf{poly}]$. In case this separation holds, it would
exhibit the largest qualitative gap between monotone and general
Boolean circuits, i.e., even extremely parallel non-monotone computations
can be more efficient than 
arbitrary
monotone computations.

\subsection{Results}

Our results show that, with respect to the computation of monotone functions, highly parallel (non-monotone) Boolean circuits can be super-polynomially more efficient than unrestricted monotone circuits. Before providing a precise formulation of these results, we introduce some notation.

For a function $d \colon \mathbb{N} \to \mathbb{N}$, let $\mDEPTH[d]$ denote the class of Boolean functions computed by monotone fan-in two $\mathsf{AND}/\mathsf{OR}$ Boolean circuits of depth $O(d(n))$. Similarly, we use $\mSIZE[s]$ to denote the class of Boolean functions computed by monotone circuits of size $O(s(n))$. More generally, for a circuit class $\mathcal{C}$, we let $\mathsf{m}\mathcal{C}$ denote its natural monotone analogue. Finally, for a Boolean function $f \colon \{0,1\}^n \to \{0,1\}$, we use $\Cmon{f}$ and $\mdepth{f}$ to denote its monotone circuit size and depth complexities, respectively. We refer to Jukna \citep{jukna_2012} for standard background on circuit complexity theory.

\subsubsection{Constant-depth circuits vs.~monotone circuits}

Recall that the  Okol'nishnikova-Ajtai-Gurevich \cite{okolnishnikova_1982,ajtai_gurevich_1987} theorem states that 
$\semMon{\AC^0} \nsubseteq \mAC^0$. In contrast, as our
main result, we establish a separation between constant-depth Boolean
circuits
and monotone circuits of much larger depth. In particular, we show that constant-depth circuits with negations can be significantly more efficient than  monotone formulas.

\begin{theorem}[Polynomial-size constant-depth vs.~larger monotone depth]
    \label{thm:ac0-fml-graph-intro}
    For every $k \geq 1$,
    we have
    $\semMon{\AC^0} \not\sseq \mDEPTH[(\log n)^k]$. Moreover, this separation holds for a monotone graph property.
\end{theorem}

In a more constrained setting, Kuperberg~\cite{DBLP:conf/lics/Kuperberg21, DBLP:journals/corr/abs-2201-11619} 
exhibited a monotone graph property expressible in first-order logic 
that cannot be expressed in positive first-order logic. A separation that holds for a monotone graph property was unknown even in the context of $\AC^0$ versus $\mAC^0$.

Let $\HomPres$ denote the class of all homomorphism-preserving graph properties, and recall that Rossman~\cite{rossman_2008, rossman_16} established that 
$\AC^0 \cap \HomPres \subseteq \mAC^0$. \Cref{thm:ac0-fml-graph-intro} implies that this efficient monotone simulation does not extend to the larger class of monotone graph properties, even if super-logarithmic depth is allowed.

Our argument is completely different from those of \citep{okolnishnikova_1982, ajtai_gurevich_1987, BST13, cos15} and their counterparts in first-order logic \citep{DBLP:conf/lics/Stolboushkin95, DBLP:conf/lics/Kuperberg21, DBLP:journals/corr/abs-2201-11619}. In particular, it allows us to break the $O(\log n)$ monotone depth barrier present in previous separations with an 
$\AC^0$ upper bound, which rely on lower bounds against monotone circuits of depth $d$ and size (at most) $2^{n^{O(1/d)}}$. We defer the discussion of our techniques to \Cref{sec:techniques}.

In our next result, we consider monotone circuits of unbounded depth.

\begin{restatable}[Polynomial-size constant-depth vs.~larger monotone size]{theorem}{acxormp}
    \label{thm:ac02-mp-intro}
    For every $k \geq 1$,
    we have
$\semMon{\AC^0[\oplus]} \not\sseq \mSIZE[2^{(\log n)^k}]$.
\end{restatable}

\Cref{thm:ac0-fml-graph-intro} and \Cref{thm:ac02-mp-intro} are incomparable: while the monotone lower bound is stronger in the latter, its constant-depth upper bound requires parity gates. \Cref{thm:ac02-mp-intro} provides the first separation between constant-depth circuits and monotone circuits of polynomial size, coming remarkably close to a solution to the question considered by Grigni and Sipser \citep{gs92}.

We note that in both of our results the family of monotone functions is explicit and has a simple description (see \Cref{sec:techniques}).

\subsubsection{Non-trivial monotone simulations and their consequences}

While \Cref{thm:ac0-fml-graph-intro} and \Cref{thm:ac02-mp-intro} provide more 
 evidence for the existence of monotone functions in $\AC^0$
which require monotone circuits of super-polynomial size, they still leave open the
intriguing possibility that
unbounded fan-in $\xor$-gates might be crucial to achieve the utmost cancellations (speedups) provided by constant-depth circuits. This further motivates the investigation of efficient monotone simulations of constant-depth circuits without parity gates, which we consider next.

For convenience, let $\AC^0_d[s]$ denote the class of Boolean functions
computed by $\AC^0$ circuits of depth $\leq d$ and size $\leq s(n)$. (We
might omit $s(n)$ and/or $d$ when implicitly quantifying over all families
of polynomial size circuits and/or all constant depths.)

We observe that a non-trivial monotone simulation is possible in the absence of parity
gates. Indeed, by combining existing results from circuit complexity theory, it is not hard to show that
$\semMon{\AC^0_d[s]} 
\sseq 
\mSIZE[2^{n(1-{1/O(\log s)^{d-1}})}]$ (see \Cref{s:simulations}).
Moreover, this upper bound is achieved by monotone $\mathsf{DNFs}$ of the same size. This is the best upper bound we can currently show for the class of all
monotone functions when the depth $d \geq 3$. (Negations offer no speedup
at depths $d \leq 2$ \citep{quine_1953}.) In contrast, we prove that a
significantly faster monotone simulation would lead to new (non-monotone)
lower bounds in  complexity theory. Recall that it is a notorious open
problem to obtain explicit lower bounds against depth-$d$ circuits of size
$2^{\omega(n^{1/(d-1)})}$, for any fixed $d \geq 3$.
We denote by $\GraphPpts$ the set of all 
Boolean functions which are graph properties.

\begin{theorem}[New circuit lower bounds from monotone simulations]
    \label{thm:simul_consequences_intro}
    There exists $\eps > 0$
    such that the following holds.
    \begin{enumerate}
        \item 
            If 
            $\semMon{\AC^0_3} 
            \subseteq \mNC^1$,
            then 
            $\NP 
            \not\sseq
            \AC^0_3[2^{o(n)}]$.
        \item 
            If 
            $\semMon{\AC^0_4} 
            \subseteq \mSIZE[\poly]$,
            then 
            $\NP 
            \not\sseq
            \AC^0_4[ 2^{o( \sqrt{n}/\log n )} ]$.
        \item If 
            $\semMon{\AC^0}  
            \sseq 
            \mSIZE[\poly]$,
            then
            $\NC^2 
            \not\sseq
            \NC^1$.
        \item 
            If $\semMon{\NC^1} \sseq \mSIZE[2^{O(n^\eps)}]$,
            then 
            $\NC^2 
            \not\sseq
            \NC^1$.
        \item 
            If 
            $\semMon{\AC^0} \cap \GraphPpts \subseteq \mSIZE[\mathsf{poly}]$,
            then $\NP \not\sseq \NC^1$.
        \item 
            If $\semMon{\NC^1} \cap \GraphPpts \sseq \mSIZE[\mathsf{poly}]$,
            then 
            $\L
            \not\sseq
            \NC^1$.
    \end{enumerate}
\end{theorem}

Item (3) of \Cref{thm:simul_consequences_intro}
implies 
in particular that,
if the 
upper bound of \Cref{thm:ac02-mp-intro}
cannot be improved to $\AC^0$ (i.e.,~the question asked by \citep{gs92} has a negative answer),
then $\NC^2 \not\sseq \NC^1$. 
It also
improves a result from \cite{choprs20} showing the weaker conclusion
$\mathsf{NP} \nsubseteq \mathsf{NC}^1$ under the same assumption.

Even if 
it's impossible to efficiently simulate
$\AC^0$ circuits computing monotone functions using unbounded depth monotone circuits,
it could still be the case that a simulation exists for certain  classes of
monotone functions with additional structure. As explained above,
Rossman's result~\cite{rossman_2008, rossman_16} achieves this for graph properties that
are preserved under homomorphisms. 
Items (5) and (6) of \Cref{thm:simul_consequences_intro}
show that a simulation that holds for all monotone 
graph properties is sufficient to get
new separations in computational complexity.

\subsubsection{Monotone complexity of constraint satisfaction problems}
\label{sec:intro_mon_csp}

Recall that 
\citep{gkrs19} 
showed the existence of a monotone function $f^{\mathsf{GKRS}}$ in $\NC^2$ that is not in $\mSIZE[2^{n^{\Omega(1)}}]$. As opposed to classical results \citep{razborov1985lower, andreev1985method, DBLP:journals/combinatorica/AlonB87, tardos_1988} that rely on the approximation method, their monotone circuit lower bound employs a  lifting technique from communication complexity. It is thus natural to consider if their approach can be adapted to provide a monotone function $g$ that is efficiently computable by constant-depth circuits but is not in $\mSIZE[\poly]$. 

As remarked in~\cite{gkrs19,DBLP:journals/sigact/RezendeGR22}, all
monotone lower bounds 
obtained from lifting theorems so far 
also hold for monotone encodings of constraint satisfaction problems (CSPs). 
Next, we introduce a class of monotone Boolean functions $\CSPf{S}$ which
capture the framework and lower bound of \citep{gkrs19}.\\

\noindent \textbf{Encoding CSPs as monotone Boolean functions.} Let $R \sseq \blt^k$ be a relation.
We call $k$ the \emph{arity} of $R$.
Let
$V = (i_1,\dots,i_k) \in [n]^k$,
and 
let
$f_{R,V}: \blt^n \to \blt$
be
the function 
that
accepts a string $x \in \blt^n$
if $(x_{i_1},\dots,x_{i_k}) \in R$.
We call $f_{R,V}$ a \emph{constraint application} of $R$ on $n$ variables.
(A different choice of the sequence $V$
gives a different constraint application of $R$.)
If $S$ is a finite set of Boolean relations,
we call any set of constraint applications of relations from $S$
on a fixed set of variables
an \emph{$S$-formula.}
In particular, we can describe an $S$-formula through a set of pairs $(V,R)$.
We say that an $S$-formula $F$ is \emph{satisfiable}
if there exists an assignment to the variables of $F$
which satisfies all the constraints of $F$.

Let $S = \set{R_1,\dots,R_k}$ be a finite set of Boolean relations.
Let $\ell_i$ be the arity of the relation $R_i$.
Note that there are $n^{\ell_i}$ possible constraint applications
of the relation $R_i$ on $n$ variables.
Let
$N := \sum_{i=1}^k n^{\ell_i}$.
We can identify each 
$S$-formula $F$
on a fixed set of $n$ variables 
with a corresponding string $w^F \in \blt^N$, where $w^F_j = 1$ if and only if the $j$-th possible constraint application (corresponding to one of the $N$ pairs $(V,R))$ appears in $F$.
Let $\CSPfn{S}{n} : \blt^N \to \blt$
be the Boolean function which accepts a given 
$S$-formula $F$ if $F$ is \emph{unsatisfiable}.
Note that this is a monotone function.
When $n$ is clear from the context or we view $\{\CSPfn{S}{n}\}_{n \geq 1}$ as a sequence of functions, we simply write
$\CSPf{S}$.\\

The function $f^{\mathsf{GKRS}}$ from \citep{gkrs19} is simply $\CSPf{S}$ for $S = \{\oplus^0_3, \oplus^1_3\}$, where $\oplus^b_3(x_1,x_2,x_3) = 1$ if and only if $\sum_i x_i = b~(\mathsf{mod}~2)$. More generally, for any finite set $S$ of Boolean relations, their framework shows how to lift a Resolution width (resp.~depth) lower bound for an arbitrary unsatisfiable $S$-formula $F$ over $m$ variables into a corresponding monotone circuit size (resp.~depth) lower bound for $\CSPfn{S}{n}$, where $n = \mathsf{poly}(m)$. 

Despite the generality of the technique from \citep{gkrs19} and the vast number of possibilities for $S$, we prove that a direct application of their approach cannot establish \Cref{thm:ac0-fml-graph-intro} and \Cref{thm:ac02-mp-intro}. This is formalised as follows. (We refer to \Cref{sec:CSPs} for much stronger forms of the result.)

\begin{theorem}[Limits of the direct approach via lifting and CSPs]
    \label{thm:cspsat-intro}
    Let $S$ be a finite set of Boolean relations. The following holds. 
    \begin{enumerate}
        \item    If 
            $\CSPf{S} \notin \mSIZE[\mathsf{poly}]$
            then
            $\CSPf{S}$ is 
             $\pL$-hard under
            $\acmred$ reductions.
        \item  
            If
            $\CSPf{S} \notin \mNC^1$
            then $\CSPf{S}$ is $\L$-hard under
            $\acmred$ reductions.
    \end{enumerate}
\end{theorem}

In particular, since there are functions (e.g.,~$\mathsf{Majority}$) computable in logarithmic space that are not in $\mathsf{AC}^0[\oplus]$, \Cref{thm:cspsat-intro} (Part 2) implies that any  $\CSPf{S}$ function that is hard for poly-size monotone formulas ($\mNC^1$) must lie outside $\mathsf{AC}^0[\oplus]$. Observe that this can also be interpreted as a \emph{monotone simulation}: for any finite set $S$ of Boolean relations, if $\CSPf{S} \in \AC^0[\oplus]$ then $\CSPf{S} \in \mNC^1$.\footnote{
Jumping ahead, our proof of \Cref{thm:ac02-mp-intro} still relies in a crucial way on the monotone lower bound obtained by \citep{gkrs19}. However, our argument requires an extra ingredient and does not follow from a direct application of their template. We provide more details about it in Section \ref{sec:techniques} below. Interestingly, the proof of \Cref{thm:ac0-fml-graph-intro} was discovered by trying to avoid the ``barrier'' posed by \Cref{thm:cspsat-intro}.}

\Cref{thm:cspsat-intro} is a corollary of a general result that completely
classifies the monotone circuit complexity of Boolean-valued constraint
satisfaction problems based on the set $\Pol(S)$ of \emph{polymorphisms} of
$S$, a standard concept in the investigation of CSPs.\footnote{Roughly
speaking, $\Pol(S)$ captures the amount of symmetry in $S$, and a larger
set $\Pol(S)$  implies that solving $\CSPf{S}$ is computationally easier.
We refer the reader to \Cref{sec:CSPs} for more details and for a
discussion of Post's lattice, which is relevant in the next statement.} We
present next a simplified version of this result, which shows a dichotomy
for the monotone circuit size and depth of Boolean-valued constraint
satisfaction problems. We refer to  \Cref{sec:CSPs} for a more general
formulation and additional consequences.

\begin{theorem}[Dichotomies for the monotone complexity of Boolean-valued CSPs]
    \label{thm:intro-mon-dichotomies}
    Let $S$ be a finite set of Boolean relations. The following holds.
    \begin{enumerate}
        \item \emph{Monotone Size Dichotomy:} If 
    $\Pol(S) \sseq \clonefont{L_3}$
     there is $\eps > 0$
    such that
    $\Cmon{\CSPf{S}} = 2^{\Omega(n^\eps)}$.
    Otherwise, 
    $\Cmon{\CSPf{S}} = n^{O(1)}$.

\item \emph{Monotone Depth Dichotomy:}
    If 
    $\Pol(S) \sseq \clonefont{L_3}$ or
    $\Pol(S) \sseq \clonefont{V_2}$ or
    $\Pol(S) \sseq \clonefont{E_2}$,
     there is $\eps > 0$
    such that
    $\mdepth{\CSPf{S}} = {\Omega(n^\eps)}$.
    Otherwise, 
    $\CSPf{S} \in \mNC^2$.
    \end{enumerate}
    
\end{theorem}

We note that previous papers of Schaefer \citep{DBLP:conf/stoc/Schaefer78} and 
Allender, Bauland, Immerman, Schnoor and Vollmer~\citep{DBLP:journals/jcss/AllenderBISV09}
provided a \emph{conditional}
classification of the complexity of such CSPs.  
\Cref{thm:intro-mon-dichotomies} and its extensions, which build on their
results and techniques, paint a complete and \emph{unconditional} picture of their
monotone complexity.\footnote{
We 
remark that only recently has Schaefer's classification
been extended to the non-Boolean
case~\cite{DBLP:conf/focs/Zhuk17,DBLP:conf/focs/Bulatov17}.
Though the refined classification of 
\cite{DBLP:journals/jcss/AllenderBISV09} 
is conjectured
to hold analogously in the case of non-Boolean
CSPs~\cite{DBLP:journals/tcs/LaroseT09},
this is still open (see the discussion
in~\cite[Section 7]{DBLP:journals/siglog/Bulatov18}).
}

\subsection{Techniques}
\label{sec:techniques}

Our arguments combine in novel ways several previously unrelated ideas from the literature. The exposition below follows the order in which the results appear above,
except for the overview of the proof of \Cref{thm:ac0-fml-graph-intro},
which appears last. We discuss this result after explaining the proof of
\Cref{thm:ac02-mp-intro} and the classification of the monotone complexity
of CSPs (\Cref{thm:cspsat-intro} and \Cref{thm:intro-mon-dichotomies}), as
this sheds light into how the proof of \Cref{thm:ac0-fml-graph-intro} was
discovered and into the nature of the argument.

\paragraph{A monotone circuit size lower bound for a function in
$\AC^0[\oplus]$.}

We first give an overview of the proof of \Cref{thm:ac02-mp-intro}.

\subparagraph{The lower bound of \emph{\cite{gkrs19}}.}

We begin by providing more details about the aforementioned monotone circuit lower bound of \cite{gkrs19}, since their result is a key ingredient in our separation (see~\cite{DBLP:journals/sigact/RezendeGR22}
for a more detailed overview). Recall that their function $f^{\mathsf{GKRS}}$ corresponds to $\CSPf{S}$ for $S = \{\oplus^0_3, \oplus^1_3\}$. Following their notation, this is simply the  Boolean function $\txorsat_n \colon \blt^{2n^3} \to \blt$ which uses each input bit to indicate the presence of a linear equation with exactly $3$ variables. This (monotone) function accepts a given linear system over $\bbf_2$
if the system is \emph{unsatisfiable}. As one of their main results, \cite{gkrs19} employed a lifting technique from communication complexity to show the existence of a constant $\varepsilon > 0$ such that $\mSIZE(\txorsat_n) = 2^{n^{\varepsilon}}$. (We show in \Cref{sec:xorsat-appr} that a weaker super-polynomial monotone circuit size lower bound for $\txorsat_n$ can also be obtained using the approximation method and a reduction.)\\

\noindent \emph{Sketch of the proof of  \Cref{thm:ac02-mp-intro}.} Since $\txorsat_n \in \NC^2$ (see, e.g,~\citep{gkrs19}), their result implies that 
$\semMon{\NC^2} \nsubseteq \mSIZE[2^{n^{\Omega(1)}}]$. On the
other hand, we are after a separation between \emph{constant-depth}
(non-monotone) circuits and \emph{polynomial-size} (unbounded depth)
monotone circuits. There are two natural ways that one might try to
approach this challenge, as discussed next.

First, the lifting framework explored by \citep{gkrs19} offers in principle the possibility that by carefully picking a different set $S$ of Boolean relations, one might be able to reduce the non-monotone depth complexity of $\CSPf{S}$ while retaining super-polynomial monotone hardness. However, \Cref{thm:cspsat-intro} shows that this is impossible, as explained above. 

A second possibility is to combine the \emph{exponential} $2^{n^\varepsilon}$ monotone circuit size lower bound for  $\txorsat_n$ and a padding argument, since we only need \emph{super-polynomial} hardness. Indeed, this argument can be used to define a monotone function $g \colon \{0,1\}^n \to \{0,1\}$ that is computed by polynomial-size fan-in two circuits of depth $\mathsf{poly}(\log \log n)$ but requires monotone circuit of size $n^{\omega(1)}$. However, it is clear that no padding argument alone can reduce the non-monotone circuit depth bound to $O(1)$ while retaining the desired monotone hardness. 

Given that both the classical widely investigated approximation method for monotone lower bounds and the more recent lifting technique do not appear to work in their current forms, for some time it seemed to us that, if true, a significantly new technique would be needed to establish a separation similar to the one in \Cref{thm:ac02-mp-intro}. 

Perhaps surprisingly, it turns out that a more clever approach that combines padding with a non-trivial circuit upper bound can be used to obtain the result. The first key observation, already present in \citep{gkrs19} and other papers, is that $\txorsat_n$ can be computed not only in $\NC^2$ but actually by polynomial-size span programs over $\bbf_2$. On the other hand, it is known that this model is equivalent in power to parity branching programs \citep{DBLP:conf/coco/KarchmerW93}, which correspond to the non-uniform version of $\pL$, i.e., counting modulo $2$ the number of accepting paths of a nondeterministic Turing machine that uses $O(\log n)$ space. A second key idea is that such a computation can be simulated by $\mathsf{AC}^0[\oplus]$ circuits of sub-exponential size and large depth. More precisely, similarly to an existing simulation of $\mathsf{NL}$ (nondeterministic logspace) by $\mathsf{AC}^0$ circuits of depth $d$ and size $2^{n^{O(1/d)}}$ via a ``guess-and-verify'' approach, it is possible to achieve an analogous simulation of $\pL$ using $\mathsf{AC}^0[\oplus]$ circuits (this folklore result appears implicit in \cite{AKRRV01} and \cite{DBLP:conf/coco/OliveiraS019}). Putting everything together, it follows that for a large enough but constant depth, $\txorsat_n$ can be computed by $\mathsf{AC}^0[\oplus]$ circuits of size $2^{n^{\varepsilon/2}}$. Since this function is hard against monotone circuits of size $2^{n^{\varepsilon}}$, a padding argument can now be used to establish a separation between $\AC^0[\oplus]$ and $\mSIZE[\poly]$. (A careful choice of parameters provides the slightly stronger statement in \Cref{thm:ac02-mp-intro}.)

\paragraph{Non-trivial monotone simulations and their consequences.}

In order to conclude that
significantly stronger monotone simulations imply
new complexity separations (\Cref{thm:simul_consequences_intro}), we argue contrapositively.
By supposing a complexity collapse,
we can
exploit known monotone circuit lower bounds
to conclude that a hard monotone function exists
in a lower complexity class.
For instance, if $\NC^2 \sseq \NC^1$, then
$\txorsat \in \NC^1$,
and we can conclude by standard depth-reduction for $\NC^1$ and padding,
together with the exponential lower bound for $\txorsat$ due
to~\cite{gkrs19},
that there exists a monotone function in $\AC^0$
which is hard for polynomial-size monotone circuits.
The other implications are argued in a similar fashion. In particular, we avoid the more complicated use of hardness magnification from \citep{choprs20} to establish this kind of result, while also getting a stronger consequence.

A little more work is required in the case of graph properties (\Cref{thm:simul_consequences_intro} Items 5 and 6),
as
padding the function computing a graph property does not yield a graph property.
We give a general lemma that allows us to pad 
monotone graph properties
while 
preserving their structure~(\Cref{lemma:padding_graph}).
We then argue as 
in the case for general functions,
using known monotone lower bounds for graph properties. We note that \Cref{lemma:padding_graph} is also important in the proof of \Cref{thm:ac0-fml-graph-intro}, which will be discussed below. We believe that our padding technique for graph properties might find additional applications.%

\paragraph{Monotone complexity of CSPs.} These are the most technical results of the paper. Since explaining the corresponding proofs requires more background and case analysis, here we only briefly describe the main ideas and references behind \Cref{thm:cspsat-intro}, \Cref{thm:intro-mon-dichotomies}, and the  extensions discussed in \Cref{sec:CSPs}.

A seminal work of Schaefer~\cite{DBLP:conf/stoc/Schaefer78}
proved that 
any Boolean CSP 
is either solvable in polynomial-time or it is $\NP$-complete.
Later, 
Jeavons~\cite{DBLP:journals/tcs/Jeavons98}
observed that
the complexity of deciding if a given set of
constraint applications of $S$ is satisfiable
depends exclusively on 
the
set 
$\Pol(S)$
of \emph{polymorphisms} of $S$.
Intuitively, the set of polymorphisms of a set of relations
is a measure of 
 its symmetry.
The more symmetric a set of relations is, 
the lesser is its expressive power.
Jeavons formally proves this intuition by showing that,
if
$\Pol(S) \sseq \Pol(S')$,
then
the problem of deciding the satisfiability of a given $S'$-formula
can be reduced in polynomial-time
to that of deciding the satisfiability of a given $S$-formula.
This allows Jeavons to reprove
Schaefer's result.

Existing proofs and classification results for constraint satisfaction problems
do
not encode the satisfiability problem
as a monotone Boolean function 
$\CSPf{S}$, in the way we described above.
We reexamine Schaefer's and Jeavons's proofs and
establish that 
the reduction from $\CSPf{S'}$ to $\CSPf{S}$
can also be done 
with efficient monotone circuits.
Making use of and adapting parts of the refined results and analysis
of 
\cite{DBLP:journals/jcss/AllenderBISV09},
which builds
on the earlier dichotomy result of \cite{DBLP:conf/stoc/Schaefer78} and
provides a detailed picture of the computational complexity of
Boolean-valued CSPs, we  prove in fact that the underlying reductions can
all be done in monotone nondeterministic logspace.

Finally, using known upper and lower bounds for monotone circuits together
with a direct analysis of some basic cases,
and inspecting Post's lattice~\cite{post41,playing_one,playing_two},
we are able to show that
$\CSPf{S}$ is hard for monotone circuits
only when
$\CSPf{S}$ is $\pL$-complete, as in \Cref{thm:cspsat-intro} Part 1.

\paragraph{A monotone circuit depth lower bound for a function in $\AC^0$.} Next, we combine insights obtained from the monotone lower bound of \cite{gkrs19}, our proof of \Cref{thm:ac02-mp-intro} via a guess-and-verify depth reduction and padding, and the statement of \Cref{thm:cspsat-intro} (limits of the direct approach via CSPs) to get the separation in \Cref{thm:ac0-fml-graph-intro}. As alluded to above, our approach differs from those of \citep{okolnishnikova_1982, ajtai_gurevich_1987, BST13,  cos15} and related results in the context of first-order logic \citep{DBLP:conf/lics/Stolboushkin95, DBLP:conf/lics/Kuperberg21, DBLP:journals/corr/abs-2201-11619}.

Recall that the \citep{gkrs19} framework lifts a Resolution width lower bound for an unsatisfiable $S$-formula $F$ into a corresponding monotone circuit size lower bound for $\CSPf{S}$. On the other hand, \Cref{thm:cspsat-intro} rules out separating constant-depth circuits from monotone circuits of polynomial size via $\CSPf{S}$ functions. In particular, we cannot directly apply the chain of reductions from \citep{gkrs19} to obtain the desired separation result. Instead, we extract from the specific $S$-formula $F$ that they use a \emph{structural property} that will allow us to improve the $\AC^0[\oplus]$ upper from \Cref{thm:ac02-mp-intro} to the desired $\AC^0$ upper bound in \Cref{thm:ac0-fml-graph-intro}.

In \citep{gkrs19} the formula $F$ is a \emph{Tseitin} contradiction, a well-known class of unsatisfiable CNFs with a number of applications in proof complexity. For an undirected graph $G$, the Tseitin formula $T(G)$ encodes a system of linear equations modulo $2$ as follows: each edge $e \in E(G)$ becomes a Boolean variable $x_e$, and each vertex $v \in V(G)$ corresponds to a constraint (linear equation) $C_v$ stating that $\sum_{u \in N_G(v)} x_{\{v,u\}} = 1~(\mathsf{mod}\;2)$, where $N_G(v)$ denotes the set of neighbours of $v$ in $G$. Crucially, $T(G)$ does not encode an arbitrary system of linear equations, i.e., the following key structural property holds: every variable $x_e$ appears in exactly $2$ equations. 

On a technical level, this property is not preserved when obtaining a (total) monotone function $\CSPf{S}$ by the gadget composition employed in the lifting framework and its reductions. However, we can still hope to explore this property in a somewhat different argument with the  goal of obtaining CSP instances that lie in a complexity class weaker than $\pL$, which is the main bottleneck in the proof of \Cref{thm:ac02-mp-intro} yielding $\AC^0[\oplus]$ circuits instead of $\AC^0$. At the same time, considering this structural property immediately takes us outside the domain of \Cref{thm:cspsat-intro}, which does not impose structural conditions over the CSP instances.

We can capture the computational problem corresponding to this type of
system of linear equations using the following Boolean function. Let
$\OddFactor_n : \blt^{\binom{n}{2}} \to \blt$ be the function that accepts
a given graph $G$ if the formula $T(G)$ described above is
\emph{satisfiable}. (Equivalently, if $G$ admits a spanning subgraph in
which the degree of every vertex is odd.) Note that $\OddFactor_n$ is a
monotone Boolean function: adding edges to $G$ cannot make a satisfiable
system unsatisfiable, since we can always set a new edge variable $x_e$ to
$0$.

While $\txorsat$ (the corresponding $\CSPf{S}$ function obtained from an
appropriate Tseitin formula via the framework of \citep{gkrs19}) admits a
$\pL$ upper bound, we observe that $\OddFactor_n$ can be computed in $\L$
thanks to its more structured class of input instances. Indeed, one can
prove that the formula $T(G)$ is satisfiable if and only if every connected
component of G has an even number of vertices.\footnote{A simple parity
argument shows that odd-sized components cannot be satisfied. On the other
hand, we can always satisfy an even-sized component by starting with an
arbitrary assignment, which must satisfy an even number of constraints by a
parity argument, and flipping the values of the edges in a path between
unsatisfied nodes, until all nodes in the connected component are
satisfied.} In turn, the latter condition can be checked in logarithmic
space using Reingold's algorithm for undirected $s$-$t$-connectivity
\cite{DBLP:conf/stoc/Reingold05}. (We note that related ideas appear in an
unpublished note of Johannsen \cite{joh03}.)
This is the first application of Reingold's algorithm
to this kind of separation.

At the same time, $\OddFactor_n$ retains at least part of the monotone hardness of $\txorsat$. Using a different reduction from a communication complexity lower bound, \cite{DBLP:journals/combinatorica/BabaiGW99} proved that the monotone circuit depth of $\OddFactor_n$ is $n^{\Omega(1)}$. Altogether, we obtain a monotone Boolean function (indeed a graph property) that lies in $\L$ but is not in $\mDEPTH[n^{o(1)}]$. Applying a guess-and-verify depth reduction for $\L$ and using (graph) padding (analogously to the proof sketch of \Cref{thm:ac02-mp-intro}), we get a monotone graph property in $\AC^0$ that is not in $\mDEPTH[\log^k n]$. This completes the sketch of the proof of \Cref{thm:ac0-fml-graph-intro}.

\subsection{Directions and open problems}\label{s:questions}

Constant-depth circuits and monotone circuits are possibly the two most widely investigated models in circuit complexity theory. Although our results provide new insights about the relation between them, there are exceptionally basic questions that remain open. 

While \cite{quine_1953} showed that negations can be efficiently eliminated from circuits of depth $d \leq 2$ that compute monotone functions, already at depth $d= 3$ the situation is much less clear. \Cref{thm:ros-dt-mdnf-intro} (see \Cref{s:simulations}) implies that every monotone function in depth-$3$ $\AC^0$ admits a monotone circuit of size $2^{n - \Omega(n/\log^2 n)}$. It is unclear to us if this is optimal. While \citep{cos15} rules out an efficient \emph{constant-depth} monotone simulation, it is still possible (and consistent with \Cref{thm:ac0-fml-graph-intro}) that 
$\semMon{\AC^0_3} \subseteq \mathsf{mNC}^1$. Is there a
significantly better monotone circuit size upper bound for monotone
functions computed by polynomial-size depth-$3$ circuits?

Our results come close to solving the question posed by Grigni and Sipser
\citep{gs92}. Using our approach, it would be sufficient to show that
$\OddFactor_n$ requires monotone circuits of size $\exp(n^{\Omega(1)})$.
This is closely related to the challenge of obtaining an exponential
monotone circuit size lower bound for $\mathsf{Matching}_n$, a longstanding
open problem in monotone complexity (see \citep[Section
9.11]{jukna_2012}).\footnote{Note that in $\OddFactor$ we are concerned
with the existence of a spanning subgraph where the degree of every vertex
is odd, while in $\mathsf{Matching}$ the degree should be exactly $1$.}
Indeed, it's possible to reduce $\OddFactor$ to $\Matching$
using monotone $\AC^0$ circuits (see~\cite[Lemma 6.18]{akiyama_factors_2011}).

Incidentally, the algebraic complexity variant of the $\mathsf{AC}^0$
vs.~$\mathsf{mSIZE}[\mathsf{poly}]$ problem has been recently settled in a
strong way through a new separation result obtained by Chattopadhyay,
Datta, and  Mukhopadhyay \cite{cdm21}. Could some of  their techniques be
useful to attack the more elusive Boolean case?

Finally, it would be interesting to develop a more general theory able to
explain when cancellations can speedup the computation of monotone Boolean
functions. Our investigation of monotone simulations and separations for
different classes of monotone functions (graph properties and constraint
satisfaction problems) can be seen as a further step in this direction.
\\

\noindent \textbf{Acknowledgements.} We thank Arkadev Chattopadhyay for
several conversations about the $\mathsf{AC}^0$ versus
$\mathsf{mSIZE}[\mathsf{poly}]$ problem and related questions. We are also
grateful to Denis Kuperberg for explaining to us the results from
\cite{DBLP:conf/lics/Kuperberg21, DBLP:journals/corr/abs-2201-11619}. The
first author thanks Ninad Rajgopal for helpful discussions about
depth reduction. 
Finally, we thank Gernot Salzer for the code used to generate Figures
\ref{fig:post_lattice}, \ref{fig:monckt-dichotomy}, and
\ref{fig:monfml-dichotomy}.
This
work received support from the Royal Society University Research Fellowship
URF$\setminus$R1$\setminus$191059, the EPSRC New Horizons Grant
EP/V048201/1, and the Centre for Discrete
Mathematics and its Applications (DIMAP) at the University of Warwick.

\section{Preliminaries}

\subsection{Notation}

\subparagraph{Boolean functions.}

We denote by $\monfns$ the set of all monotone Boolean functions.
We define $\poly = \set{n \mapsto n^C : C \in \bbn}$.
A Boolean function
    $f : \blt^{\binom{n}{2}} \to \blt$ is said to be a graph property
    if $f(G) = f(H)$ for any two isomorphic graphs $G$ and $H$.
Let $\calf = \set{f_n}_{n \in \bbn}$
    be a sequence of graph properties, where $f_n$ is defined over undirected graphs on
    $n$ vertices.
    We say that $\calf$ is
    \emph{preserved under homomorphisms} if,
    whenever there is a homomorphism from a graph $G$ to a graph $H$,
    we have $\calf(G) \leq \calf(H)$.
    We denote by $\HomPres$ 
    the set of all 
    graph properties which are preserved under homomorphisms.
    Note that $\HomPres \sseq \monfns$.

    \subparagraph{Boolean circuits.}

    We denote by $\AC^0_d[s]$ 
    the family of Boolean functions computed by
    size-$s$,
    depth-$d$
    Boolean circuits with unbounded fan-in $\set{\land,\lor}$-gates and input literals from $\{x_1, \overline{x_1}, \ldots, x_n, \overline{x_n}\}$.
    We write $\AC^0[s]$ as a shorthand for
    $\bigcup_{d = 1}^\infty \AC^0_d[s]$,
    and $\AC^0$ as a shorthand of $\AC^0[n^{O(1)}] = \AC^0[\poly]$.
    We will also refer to $\AC^0_d[\poly]$
    by $\AC^0_d$.
We write $\DNF[s]$ to denote the family of Boolean functions
    computed by size-$s$ DNFs, where size is measured by  number of terms.
    We write $\CNF[s]$ analogously.
    We write $\SIZE[s]$ to denote
    the family of Boolean functions
    computed by size-$s$ circuits.
    We write $\DEPTH[d]$ to denote
    the family of Boolean functions
    computed by fan-in 2 circuits of depth $d$.
We denote by $\AC^0[\xor]$ the family of Boolean functions computed by
polynomial-size $\AC^0$ circuits with unbounded fan-in $\xor$-gates.

We denote by $\L$ the 
family of
Boolean functions
computed by  logspace machines,
and by $\NL$ the family of Boolean functions
computed by polynomial-time nondeterministic logspace machines.
Moreover, we denote by $\pL$
the family of Boolean functions computed by
polynomial-time nondeterministic logspace machines with a \emph{parity}
acceptance condition (i.e.,
an input is accepted if the number of accepting paths is odd).

\subparagraph{Circuit complexity.}

Given a circuit class $\calc$, we write
    $\monclass{C}$ to denote the
    monotone version of $\calc$.
Given a function $f$, we write $\Cmon{f}$
    to denote the size of the smallest monotone circuit computing
    $f$ and $\mdepth{f}$ to denote the smallest depth of a fan-in 2
    monotone circuit computing $f$.
Given two Boolean functions $f,g$, we write
$f \mprojred g$ if
there exists a many-one reduction from $f$ to $g$
in which each bit of the reduction is a monotone
projection\footnote{A monotone projection is a projection without
negations.}
of the input.

\subparagraph{Miscellanea.}

Let $\alpha \in \blt^n$ .
    We 
    define
    $\hwt{\alpha} := \sum_{i=1}^n \alpha_i$.
    We call $\card{\alpha}_1$ the \defx{Hamming weight} of $\alpha$.
    We let
    $\supp(\alpha) = \set{i \in [n] : \alpha_i = 1}$.
We let $\thr_{k,n} : \blt^n \to \blt$ be the Boolean function
    such that
    $\thr_{k,n}(x) = 1 \iff \hwt{x} \geq k$.

\subsection{Background results}

The next lemma, which is proved via a standard ``guess-and-verify''
approach, shows that nondeterministic logspace computations can be
simulated by circuits of size $2^{n^{\varepsilon}}$ and of depth $d =
O_\varepsilon(1)$.

\begin{lemma}[Folklore; see, e.g.,~{\cite[Lemma 8.1]{AHMPS08}}]
    \label{lemma:ac0_nc1_sim}
    For all $\eps > 0$, we have
    $\NL \sseq \AC^0[2^{n^\eps}]$.
\end{lemma}

\section{Constant-Depth Circuits vs.~Monotone Circuits}
\label{s:constant-depth}

In this section, we prove \Cref{thm:ac02-mp-intro,thm:ac0-fml-graph-intro}.
For the upper bounds, we require
the logspace graph connectivity algorithm due
to~\cite{DBLP:conf/stoc/Reingold05}
and the $\pL$ algorithm for solving linear systems over $\mathbb{F}_2$ due
to~\cite{DBLP:journals/mst/BuntrockDHM92},
as well as the
depth-reduction techniques of~\cite{AKRRV01,AHMPS08}.
On the lower bounds side,
our proofs rely on previous monotone circuit and depth lower bounds
from~\cite{DBLP:journals/combinatorica/BabaiGW99,gkrs19}.
In order to obtain a monotone formula lower bound for a graph property,
we prove a graph padding lemma in \Cref{s:padding_graph}.

\subsection{A monotone size  lower bound for a function in $\AC^0[\oplus]$}
\label{s:ac0-mp}

In this section, we prove \Cref{thm:ac02-mp-intro}.
We first recall the monotone circuit lower bound of~\cite{gkrs19}
and a depth-reduction lemma implicit in~\cite{AKRRV01} and \citep{DBLP:conf/coco/OliveiraS019},
whose full proof we give below for completeness.
We remark that similar arguments can be employed
to prove \Cref{lemma:ac0_nc1_sim},
essentially by replacing the $\xor$ gates by $\lor$ gates.

As explained in \cref{sec:techniques}, in its strongest form the separation result from \citep{gkrs19} can be stated as follows.

\begin{theorem}[\cite{gkrs19}]
    \label{thm:xor-sat}
    There exists $\eps > 0$ 
    such that 
    $\semMon{\pL} \not\sseq \mSIZE[2^{o(n^{\eps})}]$.
    Moreover, 
    this separation is witnessed by $\txorsat$.
\end{theorem}

\begin{lemma}[Folklore; see, e.g.,~\cite{AKRRV01, DBLP:conf/coco/OliveiraS019}]
    \label{lemma:pl-to-ac02}
    Let $f : \blt^n \to \blt$
    be a Boolean function computed by
    a $\pL$ machine.
    For every $\delta > 0$,
    there exists an $\AC^0[\oplus]$ circuit
    of size $2^{n^\delta}$ that computes $f$.
\end{lemma}
\begin{proof}
    Let $M$ be a $\pL$-machine computing $f$.
    Without loss of generality, we may assume that
    each configuration in the 
    configuration graph $G$ of $M$
    is \emph{time-stamped} -- in other words,
    each configuration 
    carries the information of the number of
    computational steps it takes to arrive at it.\footnote{Formally,
        we can define a $\pL$-machine $M'$ such that
        the configurations of $M'$
        are
        $(C,t)$, where
        $C$ is a configuration of $M$, and $t = 0,1,\dots,m = n^{O(1)}$
        is a number denoting the time in which the configuration was
        achieved.
        A configuration $(C,t)$ can only reach a configuration
        $(C',t+1)$ in the configuration graph of $M'$.
    }
    We may also assume that
    every accepting computation takes exactly the same amount of time,
    which means that every path from the starting configuration
    $v_{\mathsf{start}}$
    to the accepting configuration
    $v_{\mathsf{accept}}$
    has the same length in the configuration graph.
    These assumptions imply that the configuration graph is \emph{layered} 
    (because a configuration with time-stamp $t$
        can only 
        point to
    configurations with time-stamp $t+1$)
    and acyclic.
    Note that, for a fixed machine,
    the configuration graph can be computed from the input string using a projection.  

    Let $m = n^{O(1)}$ be the time that an accepting computation takes.
    We now show
    how to count (modulo 2)
    the number of accepting paths from
    $v_{\mathsf{start}}$
    to
    $v_{\mathsf{accept}}$ with a depth-$d$ $\AC^0[\oplus]$ circuit.
    First, choose
    $m^{1/d}-1$ configurations $v_1,\dots,v_{m^{1/d}-1}$
    (henceforth called ``checkpoints'')
    from $V(G)$, such that the configuration
    $v_i$ is at the level
    $i \cdot m^{1-1/d}$ in the configuration graph
    (i.e., it takes $i \cdot m^{1-1/d}$ time steps to arrive at $v_i$).
    For convenience, we let $v_0 = v_{\mathsf{start}}$
    and $v_{m^{1/d}} = v_{\mathsf{accept}}$.
    We then count the number of paths from 
    from $v_{\mathsf{start}}$
    to $v_{\mathsf{accept}}$
    that go through $v_1,\dots,v_{m^{1/d}-1}$,
    and sum over all possible choices of the checkpoints.
    Since the graph is layered and
    each path from $v_0$ to $v_{m^{1/d}}$ has length exactly $m$,
    there is only one choice of checkpoints that witnesses
    a given path from $v_0$ to $v_{m^{1/d}}$,
    so no path is counted twice in this summation.
    Letting
    $\#\mathsf{paths}(s,t,\ell)$ denote the number of paths between
    configurations $s$ and $t$ with distance exactly $\ell$,
    we obtain
    \begin{equation*}
        \#\mathsf{paths}(v_{0},v_{m^{1/d}}, m)
        =
        \sum_{v_1,\dots,v_{m^{1/d}-1}}
        \prod_{i=0}^{m^{1/d}-1}
        \#\mathsf{paths}(v_i,v_{i+1}, m^{1-1/d}).
    \end{equation*}
    The above calculation can be done in modulo 2 with an unbounded fan-in XOR gate
    (replacing the summation)
    and an unbounded fan-in AND gate
    (replacing the product).
    Note that the formula above is recursive.
    Repeating the same computation for calculating (modulo 2)
    the expression
    $\#\mathsf{paths}(v_i,v_{i+1}, m^{1-1/d})$ for each $i$,
    we obtain a depth-$2d$ $\AC^0[\oplus]$ circuit
    for calculating the number of paths from $v_{\mathsf{start}}$ to
    $v_{\mathsf{accept}}$ (modulo 2).
    Clearly, the total size of the circuit is
    $2^{O(m^{1/d} \cdot \log m)}$,
    which is smaller than $2^{n^\delta}$ for a large enough constant $d$.
\end{proof}

We now restate \Cref{thm:ac02-mp-intro} and prove it
by combining
\Cref{thm:xor-sat,lemma:pl-to-ac02}
with a padding trick.

\acxormp*
\begin{proof}
    By \Cref{thm:xor-sat}, there exists $\eps > 0$ and 
    a monotone function $f \in \pL$
    such that any monotone circuit computing $f$ has size
    $2^{\Omega(n^\eps)}$.

    Let $\delta = \eps/k$ and let $m = 2^{n^{\delta}}$.
    Let $g : \blt^{n} \times \blt^{m} \to \blt$
    be the Boolean function defined as
    $g(x,y) = f(x)$.
    Note that $g$ is a function on $N := m + n = 2^{\Theta(n^{\delta})}$
    bits.
    By \Cref{lemma:pl-to-ac02},
    there exists an $\AC^0[\oplus]$ circuit computing
    $f$ of size $2^{n^{\delta}} = N^{O(1)}$. The same circuit computes $g$.
    On the other hand, any monotone circuit computing $g$
    has size
    $2^{\Omega(n^\eps)} = 2^{\Omega((\log N)^{\eps / \delta})} =
    2^{\Omega((\log N)^k)}$.
\end{proof}

\subsection{A monotone depth lower bound for a graph property in $\AC^0$}
\label{s:padding_graph}

In this section,
we prove 
\Cref{thm:ac0-fml-graph-intro}.
We prove moreover that the function that separates $\semMon{\AC^0}$ and
$\mNC^i$ can be taken to be a graph property.
We state our result in its full generality below.

\begin{theorem}
    \label{thm:ac0-fml-graph}
    For every $i \geq 1$,
    we have
    $\semMon{\AC^0} \cap \GraphPpts \not\sseq \mDEPTH[(\log n)^i]$.
    In particular, we have
    $\semMon{\AC^0} \cap \GraphPpts \not\sseq \mNC^i$.
\end{theorem}

First, we recall a result
of~\cite{DBLP:journals/combinatorica/BabaiGW99},
which proves monotone lower bounds for the following function.
Let 
$\OddFactor_n : \blt^{\binom{n}{2}} \to \blt$
be the function that accepts a given 
graph
if 
it
contains an \emph{odd factor} -- in other words,
a spanning subgraph in which the degree of every vertex is odd.
Babai, Gál and Wigderson~\cite{DBLP:journals/combinatorica/BabaiGW99}
proved
the following result:
\begin{theorem}[\cite{DBLP:journals/combinatorica/BabaiGW99}]
    \label{thm:bip-oddfactor-lb}
        Any monotone formula computing 
        $\OddFactor_n$ 
        has size $2^{\Omega(n)}$,
        and any monotone circuit computing
        $\OddFactor_n$ 
        has size $n^{\Omega(\log n)}$.
\end{theorem}

The proof in \cite{DBLP:journals/combinatorica/BabaiGW99}  is actually for the case of
\emph{bipartite graphs}, but it easily extends to general graphs, since
the bipartite case reduces to the general case by a monotone projection.
The formula lower bound stated above is slightly stronger because it makes
use of asymptotically optimal lower bounds on the randomized communication
complexity of $\mathsf{DISJ}_n$
\citep{DBLP:journals/siamdm/KalyanasundaramS92}, which were not available
to \cite{DBLP:journals/combinatorica/BabaiGW99}. We remark that, with a
different language,
a monotone circuit lower bound for $\OddFactor$ is also implicitly proved
in Feder and
Vardi~\cite[Theorem 30]{DBLP:journals/siamcomp/FederV98}.

We now recall an upper bound for $\OddFactor$,  implicitly proved
in an unpublished note due to Johannsen~\cite{joh03}.

\begin{theorem}[\cite{joh03}]
    \label{thm:oddfactor-logspace}
    We have $\OddFactor \in \L$.
\end{theorem}
\begin{proof}
    We first recall the following observation about the $\OddFactor$
    function, which appears in different forms in the literature
    (see~\cite[Lemma 4.1]{DBLP:journals/jacm/Urquhart87}
    or~\cite[Lemma 18.16]{jukna_2012};
    see also~\cite[Proposition 1]{joh03} for a different proof.)
    \begin{claim*}
        A graph $G$ has an odd factor if and only if
        every connected component of $G$ has an even number of vertices.
    \end{claim*}
    \begin{proof}
        If a graph $G$ has an odd factor, 
        we can conclude
        that every
        connected
        component of $G$ has an even number of vertices
        from the well-known observation
        that in every graph 
        there is an even number of vertices of odd degree.

        Now suppose that every connected component of $G$ 
        has an even number of vertices.
        We will iteratively construct an odd factor $F$ of $G$.
        We begin with the empty graph.
        We take any two vertices $u,v$ in the same
        connected component of $G$
        which currently have even degree in $F$,
        and consider any path
        $P = (x_1,\dots,x_{k})$
        between $u$ and $v$,
        where $x_1 = u$ and $x_{k}=v$.
        If the edge $x_i x_{i+1}$ is currently in $F$, we remove 
        $x_i x_{i+1}$ from $F$; otherwise, we add $x_i x_{i+1}$ to
        $F$.
        It's easy to check that, in every iteration of this procedure,
        only the vertices $u$ and $v$ have the parity of their degree
        changed in $F$; the degree of every other vertex stays the same
        (modulo 2).
        Since every connected component has an even number of vertices,
        this means that, eventually, every vertex in $F$ will have odd
        degree.
    \end{proof}
    Now it's easy to check in logspace if every connected component of $G$ has an even
    number of vertices using Reingold's algorithm for undirected
    connectivity~\cite{DBLP:conf/stoc/Reingold05}.
    It suffices to check if, for every vertex $v$ of $G$,
    the number of vertices reachable from $v$ is odd.
\end{proof}

Now, if we only desire to obtain a function in $\AC^0$
not computed by monotone circuits of depth $(\log n)^i$,
we can follow the same argument of~\Cref{thm:ac02-mp-intro},
using \Cref{lemma:ac0_nc1_sim} instead of \Cref{lemma:pl-to-ac02}.
In order to obtain moreover a monotone \emph{graph property} witnessing
this separation, we will need the following lemma,
which enables us to obtain a graph property after ``padding'' a graph property.
We defer the proof of this lemma to
the end of this section.

\begin{restatable}{lemma}{paddinggraph}
    \label{lemma:padding_graph}
    Let $f : \blt^{\binom{n}{2}} \to \blt$
    be a monotone graph property on graphs of $n$ vertices.
    The following holds.
    \begin{enumerate}
        \item If $f \in \NC^i$ for some $i > 1$, then
            there exists a monotone graph property 
            $g$ on graphs of 
            $N = 2^{(\log n)^i}$ vertices
            such that 
            $g \in \NC^1$
            and 
            $f \mprojred g$.
        \item If $f \in \NL$, 
            then
            for all $\eps > 0$
            there exists
            a monotone graph property 
            $g$ on graphs of 
            $N = 2^{n^\eps}$ vertices
            such that 
            $g$ can be computed by  $\AC^0$ circuits of size $N^{2+o(1)}$
            and 
            $f \mprojred g$.
        \item If $f \in \pL$, 
            then
            for all $\eps > 0$
            there exists
            a monotone graph property 
            $g$ on graphs of 
            $N = 2^{n^\eps}$ vertices
            such that 
            $g$ can be computed by  $\AC^0[\xor]$ circuits of size $N^{2+o(1)}$
            and 
            $f \mprojred g$.
    \end{enumerate}
\end{restatable}

We are now ready to prove
\Cref{thm:ac0-fml-graph}.
\begin{proof}[Proof of \Cref{thm:ac0-fml-graph}]
            Fix $n \in \bbn$ and take an $\eps < 1/i$.
            Observing that $\L \sseq \NL$,
            from
            \Cref{thm:oddfactor-logspace}
            and item (2) of \Cref{lemma:padding_graph}
            we conclude that
            there exists a monotone graph property $f$
            on 
            $N = 2^{n^\eps}$
            vertices
            such that 
            $f \in \AC^0$
            and $\OddFactor_n \mprojred f$.
            By \Cref{thm:bip-oddfactor-lb},
            any monotone circuit computing $f$
            has depth 
            $\Omega(n) = \Omega((\log N)^{1/\eps}) \gg (\log N)^i$.
\end{proof}

Raz and Wigderson~\cite{RW92} observed that there exists a monotone function
$f \in \NC^1 \sm \mNC$.
Using~\Cref{lemma:padding_graph},
we observe moreover that it's possible to obtain this separation
with a monotone graph property.

\begin{proposition}
    \label{prop:nc1-mnc-graph}
    We have $\semMon{\NC^1} \cap \GraphPpts \not\sseq \mNC$.
\end{proposition}
\begin{proof}
    Observing that $\L \sseq \NC^2$,
    we conclude 
    from
    \Cref{thm:oddfactor-logspace}
    and item (1) of \Cref{lemma:padding_graph}
    that
    there exists a monotone graph property $f$
    on 
    $N = 2^{(\log n)^2}$ 
    vertices
    such that 
    $f \in \NC^1$
    and 
    $\OddFactor_n \mprojred f$.
    By 
    \Cref{thm:bip-oddfactor-lb},
    any monotone circuit computing $f$
    has depth 
    $\Omega(n) = \Omega(2^{\sqrt{\log N}})$,
    which implies $f \not\in \mNC$.
    \qedhere
\end{proof}

\subsection{Efficient monotone padding for graph properties}

We will now prove \Cref{lemma:padding_graph}.
We first recall some low-depth circuits
for computing threshold functions, which we will use to design a circuit for efficiently computing
the adjacency matrix of induced subgraphs.

\begin{theorem}[\cite{hwwy94}]
    \label{thm:polylog-thr-ac0}
Let $d > 0$ be a constant.
The function
$\thr_{(\log n)^d,n}$ 
can be computed by an $\AC^0$ circuit of size $n^{o(1)}$
and depth $d+O(1)$.
\end{theorem}

\begin{theorem}[\cite{aks83}]
    \label{thm:thr-nc1}
    For every $k \in [n]$,
    the function $\thr_{k,n}$ can be computed by a
    circuit of depth $O(\log n)$ and size
    $n^{O(1)}$.
\end{theorem}

\begin{lemma}
    \label{lemma:induced_subgraph_circuit}
    There exists a circuit $C^k_n$ with $\binom{n}{2} + n$ inputs
    and $\binom{k}{2}$ outputs which,
    when given as input
    an adjacency matrix 
    of a graph $G$ on $n$ vertices
    and a characteristic vector of a set $S \sseq [n]$
    such that $\card{S} \leq k$,
    outputs the adjacency matrix of the graph $G[S]$,
    padded with isolated vertices when $\card{S} < k$.
    The circuit has constant-depth and size $n^{2+o(1)}$
    when $k = \polylog(n)$, and
    size $n^{O(1)}$ and depth $O(\log n)$ otherwise.
\end{lemma}
\begin{proof}
    Let $\set{x_{ij}}_{i,j \in [n]}$ encode the adjacency matrix of $G$.
    Let $\alpha \in \blt^n$ be the characteristic vector of $S$.
    Let $i,j \in [k]$.
    Note that $\set{i,j} \in E(G[S])$
    if and only if
    there exists
    $a,b \in [n]$
    such that
    \begin{itemize}
        \item 
            $\alpha_a$ is the $i$-th non-zero entry of $\alpha$,
        \item 
            $\alpha_b$ is the $j$-th non-zero entry of $\alpha$, and
        \item 
            $x_{ab} = 1$ (i.e., $a$ and $b$ are connected in $G$).
    \end{itemize}
    We first consider the case $k = \polylog(n)$.
    In this case,
    the first two conditions can be checked with 
    circuits of size
    $n^{o(1)}$
    using \Cref{thm:polylog-thr-ac0}.
    Therefore, we can 
    compute 
    if $i$ and $j$ are adjacent
    using $n^{2+o(1)}$ gates and constant depth.
    As there are at most 
    $(\log n)^{O(1)}$ 
    such pairs,
    we can output $G[S]$
    with at most 
    $n^{2+o(1)}$ gates.

    For any $k$, the first two conditions can be checked with an $\NC^1$
    circuit by Theorem~\ref{thm:thr-nc1}.
    Since there are at most $n^2$ pairs $i,j$, the entire adjacency matrix
    can be computed with a $O(\log n)$-depth and polynomial-size circuit.
\end{proof}

We are ready to prove \Cref{lemma:padding_graph}.

\begin{proof}[Proof of \Cref{lemma:padding_graph}]
    We first prove (1). 
    Fix $n \in \bbn$
    and let $N = 2^{(\log n)^i}$.
    For a graph $G$ on $N$
    vertices such that 
    $\card{E(G)} \leq \binom{n}{2}$,
    let
    $\Gclean{G}$ be the graph 
    obtained
    from $G$ by removing isolated vertices from $G$ one-by-one,
    in lexicographic order,
    until one of the following two 
    conditions are satisfied:
    (1) there are no more isolated vertices in $\Gclean{G}$,
    \emph{or}
    (2) $\Gclean{G}$ has exactly $n$ vertices.
    Let $g : \blt^{\binom{N}{2}} \to \blt$
    be the monotone graph property defined as follows:
    \begin{align*}
        g(G) := 
        \left(
            \card{E(G)} > 
            \binom{n}{2}
        \right)
        \lor
        \left(
            \card{V(\Gclean{G})} > n
        \right)
        \lor
            (f(\Gclean{G}) = 1).
    \end{align*}
    Note that $g$ accepts a graph $G$
    if and only if at least 
    one of the following three conditions are satisfied:
    \begin{enumerate}
        \item 
            $G$ has at most $\binom{n}{2}$ edges,
            $\Gclean{G}$ has exactly $n$ vertices
            and 
            $f(\Gclean{G}) = 1$,
            or
        \item 
            $G$ has more than $\binom{n}{2}$ edges,
            or
        \item 
            $\Gclean{G}$ has more than $n$ vertices.
    \end{enumerate}
    We observe that the monotonicity of $g$ follows from 
    the monotonicity of $f$.
    We also claim that $g$ is a graph property.
    Indeed, the graph $\Gclean{G}$ is the same (up to isomorphism),
    irrespective  of the order according to which the isolated vertices are
    removed from $G$. Moreover, the function $f$ is also a graph property.
    Because of this, all the three conditions above
    are preserved under isomorphisms.

    We first observe that $f$ is a monotone projection of $g$.
    Indeed, given a graph $G$ on $n$ vertices,
    we can easily construct by a monotone
    projection a graph $G'$ on $N$ vertices
    and at most $\binom{n}{2}$ edges
    such that
    $f(G) = g(G')$.
    We just let $G'$ have a planted copy of $G$,
    and all other vertices are isolated.
    Then $\Gclean{G'} = G$ (up to isomorphism)
    and $g(G') = f(\Gclean{G}) = f(G)$.

    We now show how to compute $g$ in $\NC^1$.
    Let $\set{x_{ij}}_{i,j \in [N]}$ be the input bits of $g$,
    corresponding to the adjacency matrix of a graph $G$.
    The circuit computes as follows.

    \begin{enumerate}
        \item If $\card{E(G)} > \binom{n}{2}$, accept the graph $G$.
        \item Compute the characteristic vector $\alpha \in \blt^N$
            of the set of all non-isolated vertices of $G$.
            If $\hwt{\alpha} > n$, accept the graph $G$.
        \item Compute $\Gclean{G}$ and output $f(\Gclean{G})$.
    \end{enumerate}

    Note that checking if
    $\card{E(G)} > \binom{n}{2}$
    can be done 
    in $\NC^1$ by Theorem~\ref{thm:thr-nc1}.
    Moreover,
    for all $i \in [N]$,
    we have $\alpha_i = \bigvee_{j \in [N]} x_{ij}$,
    and therefore $\alpha_i$ can be computed by a circuit of depth $O(\log N)$ 
    and $O(N)$ gates.
    In total, 
    the vector $\alpha$ can be computed
     with $O(N^2)$ gates and $O(\log N)$ depth.
    Finally, we can check if
    $\hwt{\alpha} > n$
    in $\NC^1$ with a threshold circuit.

    For the final step, we compute $\Gclean{G}$.
    If $\hwt{\alpha} = n$,
    note that $\Gclean{G} = G[\supp(\alpha)]$.
    When $\hwt{\alpha} < n$,
    then $\Gclean{G}$ is $G[\supp(\alpha)]$ padded with isolated vertices.
    We can therefore compute $\Gclean{G}$
    with the circuit 
    $C_N^n$
    of Lemma~\ref{lemma:induced_subgraph_circuit}.
    Moreover, since $f \in \NC^i$,
    we have that $f$ can be computed by a circuit of size
    $n^{O(1)} = N^{o(1)}$
    and depth
    $O((\log n)^i) = O(\log N)$.
    Therefore, computing $f(\Gclean{G})$ can be done in $\NC^1$. Overall, we get that $g \in \NC^1$.

    In order to prove (2), it suffices to modify the proof above.
    The modification can be briefly described as follows.
    We let $N = 2^{n^\eps}$.
    Every time \Cref{lemma:induced_subgraph_circuit} is applied,
    we use the $\AC^0$ circuit instead of the $\NC^1$ circuit,
    since $n = \polylog(N)$. This ammounts to $N^{2+o(1)}$ many gates with
    unbounded fan-in.
    Moreover, since by assumption 
    $f \in \NL$,
    applying \Cref{lemma:ac0_nc1_sim} we obtain 
    an
    $\AC^0$ circuit for
    $f$ of size
    $2^{n^{\eps/2}} = N^{o(1)}$,
    so we can compute $f(\Gclean{G})$ in constant depth
    with $N^{o(1)}$ gates.

    Finally, for (3) it suffices to apply the same argument used for (2),
    replacing an application of~\Cref{lemma:ac0_nc1_sim}
    by an application of \Cref{lemma:pl-to-ac02}.
\end{proof}

\section{Non-Trivial Monotone Simulations and Their Consequences}
\label{s:consequences}

In contrast to \Cref{s:constant-depth},
in this section we observe that a non-trivial simulation of $\AC^0$ circuits
by monotone circuits
is possible.
This follows from a refined version of the switching lemma proved by
Rossman~\cite{rossman_entropy}.
As a proof of concept, we use this simulation result to reprove a well-known
$\AC^0$ lower bound for $\maj$.

In the second part of this section,
we show that if much faster simulations are possible,
then even stronger non-monotone circuit lower bounds follow.
We also show that this implication is true
even if the simulation only holds for \emph{graph properties}.
Monotone simulations for graph properties are motivated by a result
of Rossman~\cite{rossman_2008}, which shows that
very strong monotone simulations are possible for
\emph{homomorphism-preserving graph properties}.
The 
lower bounds
from monotone simulations 
are proved with
the simulation result and
padding argument used in the previous section
(\Cref{lemma:ac0_nc1_sim,lemma:padding_graph}).

\subsection{A non-trivial simulation for bounded-depth circuits}
\label{s:simulations}

The earliest monotone simulation result was proved for $\DNF$s  by Quine~\cite{quine_1953}.

\begin{theorem}[Quine~\cite{quine_1953}]
    \label{t:quine_dnf}
    For all $s : \bbn \to \bbn$, we have
    $\semMon{\DNF[s]} \sseq \mDNF[s]$.
\end{theorem}
\begin{proof}
    If a given $\DNF$ computes a monotone Boolean function,
    simply removing the negative literals continues to compute the same
    function.
\end{proof}

Let $\dtsize{f}$ denote the size of a smallest decision-tree computing
$f$. We will need a result obtained by
Rossman~\cite{rossman_entropy}.

\begin{theorem}[\cite{rossman_entropy}]
    \label{t:ros-dt}
    If $f : \blt^n \to \blt$ is computable by an $\AC^0$ circuit of depth
    $d$ and size $s$,
    then 
    $\dtsize{f} = 
    2^{(1-{1/O(\log s)^{d-1}})n}$.
\end{theorem}

\begin{theorem}
    \label{thm:ros-dt-mdnf-intro}
    Let $s : \bbn \to \bbn$ and $d \geq 1$.
    We have
    $\semMon{\AC^0_d[s]} 
    \sseq 
    \mSIZE[t]$, 
    where
    $t = n \cdot 2^{n(1-{1/O(\log s)^{d-1}})}$.
    Moreover, this upper bound is achieved by monotone $\mathsf{DNFs}$ of 
    size $t/n$.
\end{theorem}
\begin{proof}
    Let $f$ be a monotone function
    computable
    by an $\AC^0$ circuit of depth
    $d$ and size $s$.
    By \Cref{t:ros-dt},
    there exists a decision tree of size
    $2^{(1-{1/O(\log s)^{d-1}})n}$
    computing $f$.
    Therefore, there exists a $\DNF$ of the same size
    computing $f$,
    which can be taken to be monotone by
    by \Cref{t:quine_dnf}.
    This can be converted into a monotone circuit of size
    $n \cdot 2^{(1-{1/O(\log s)^{d-1}})n}$.
\end{proof}

We observe that it is possible to immediately deduce an $\AC^0$ lower bound
for $\maj$ using this simulation theorem.
Even though near-optimal lower bounds for $\maj$
have been known for a long time~\cite{DBLP:conf/stoc/Hastad86} 
and the proof of the main technical tool (\Cref{t:ros-dt}) behind our simulation result is
similar to the one used by~\cite{DBLP:conf/stoc/Hastad86},
the argument below illustrates how a monotone simulation can lead to non-monotone circuit lower bounds.

\begin{corollary}
    \label{cor:maj-lb}
    Any depth-$d$ $\AC^0$ circuit computing $\maj$
    has size $2^{\Omega((n/\log n)^{1/(d-1)})}$.
\end{corollary}
\begin{proof}
    Note that $\maj$ has $\binom{n}{n/2} = \Omega(2^n/\sqrt{n})$ minterms.
    Therefore, any monotone $\DNF$ computing $\maj$ has
    size at least $\Omega(2^n/\sqrt{n})$.
    By \Cref{thm:ros-dt-mdnf-intro}, it follows that the size $s$ of a depth-$d$
    $\AC^0$ computing $\maj$ satisfies the following inequality:
    \begin{equation*}
        2^{n(1-{1/O(\log s)^{d-1}})}
        =
        \Omega(2^{n-\frac{1}{2}\log n}).
    \end{equation*}
    From this equation we obtain
    $s = 2^{\Omega((n/\log n)^{1/(d-1)})}$.
\end{proof}

\subsection{Non-monotone lower bounds
from monotone simulations}

We now show that if monotone circuits are able to
efficiently simulate non-monotone circuits computing monotone Boolean
functions,
then striking complexity separations follow.
We also show a result of this kind for simulations of graph properties.
We first prove a lemma connecting the simulation of $\AC^0$ circuits
with the simulation of $\NL$ machines.

\begin{lemma}
    \label{thm:nc1-ac0-size}
    For all constants $\eps > 0$ and $C \geq 1$, 
    if 
    $\semMon{\AC^0} 
    \sseq \mSIZE[2^{O((\log n)^C)}]$, 
    then
    $\semMon{\NL} \sseq \mSIZE[2^{o(n^\eps)}]$.
\end{lemma}
\begin{proof}
    We prove the contrapositive.
    Suppose that there exists $\eps > 0$
    such that
    $\semMon{\NL} \not\sseq \mSIZE[2^{o(n^\eps)}]$.
    This means that there exists a monotone function $f$
    such that 
    $f \in \NL$
    and any monotone circuit computing $f$ has size $2^{\Omega(n^\eps)}$.

    Let $\delta = \eps/(2C)$ and let $m = 2^{n^{\delta}}$.
    Let $g : \blt^{n} \times \blt^{m} \to \blt$
    be the Boolean function defined as
    $g(x,y) = f(x)$.
    Note that $g$ is a function on $N := m + n = 2^{\Theta(n^{\delta})}$
    bits.
    By \Cref{lemma:ac0_nc1_sim},
    there exists an $\AC^0$ circuit computing
    $f$ of size 
    $2^{n^\delta} = N^{O(1)}$.
    Moreover, any monotone circuit computing $g$
    has size
    $2^{\Omega(n^\eps)} = 2^{\Omega((\log N)^{\eps / \delta})} =
    2^{\Omega((\log N)^{2C})}$.
\end{proof}

Next, we recall the strongest known monotone circuit 
and formula
lower 
bounds
for a
monotone function in $\NP$. 

\begin{theorem}[\cite{DBLP:conf/stoc/PitassiR17}]
    \label{thm:pr17-fnp}
     $\semMon{\NP} \not\sseq \mDEPTH[o(n)]$.
\end{theorem}

\begin{theorem}[\cite{CKR20}]
    \label{thm:ckr20}
     
    $\semMon{\NP} \not\sseq \mSIZE[2^{o(\sqrt{n}/\log n)}]$.
\end{theorem}

We are now ready to state and prove our first result regarding
new complexity separations from monotone simulations. Recall that obtaining explicit lower bounds against depth-$3$ $\AC^0$  circuits of size $2^{\omega(n^{1/2})}$ is a major challenge in circuit complexity theory, while the best lower bound on the size of depth-4 $\AC^0$ circuits
computing a function in $\NP$ is
currently $2^{\Omega(n^{1/3})}$~\cite{DBLP:conf/stoc/Hastad86}.
Moreover, no strict separation is known in the following sequence of inclusions
of complexity classes:
$\ACC \sseq \TC^0 \sseq \NC^1 \sseq \L \sseq \NL \sseq \pL \sseq \NC^2$.
We show that 
efficient monotone simulations would bring new results in both of these
fronts. 
(We stress that all lower bound consequences appearing below refer to
separations against non-uniform circuits.)\footnote{In other words, all upper bounds are
\emph{uniform}, but the lower bounds hold even for \emph{non-uniform}
circuits. Note that this is stronger than lower bounds for uniform
circuits.}

\begin{theorem}
    \label{cor:xor-sat-ag}
    Let $\calc$ be 
    a class of circuits.
    There exists $\eps > 0$ such that
    the following holds:
    \begin{enumerate}
        \item 
            \label{item:ac03}
            If $\semMon{\AC^0_3} \subseteq \mNC^1$,
            then 
            $\NP 
            \not\sseq
            \AC^0_3[2^{o(n)}]$.
        \item 
            \label{item:ac04}
            If 
            $\semMon{\AC^0_4} \subseteq \mSIZE[\poly]$,
            then 
            $\NP 
            \not\sseq
            \AC^0_4[ 2^{o( \sqrt{n}/\log n )} ]$.
        \item 
            \label{item:subexp}
            If $\semMon{\calc} \sseq \mSIZE[2^{O(n^\eps)}]$,
            then 
            $\NC^2 
            \not\sseq
            \calc$.
        \item 
            \label{item:ac-poly}
            If 
            $\semMon{\AC^0} 
            \sseq 
            \mSIZE[\poly]$,
            then
            $\NC^2 
            \not\sseq
            \NC^1$.
    \end{enumerate}
\end{theorem}
\begin{proof}
    We will prove each item separately.
    \begin{enumerate}[label=\emph{Proof of (\arabic*)}., leftmargin=*]
        \item 
            Let us assume that
            $\semMon{\AC^0_3} \subseteq \mNC^1$.
            Let $f$ be the function of Theorem~\ref{thm:pr17-fnp}.
            For a contradiction,  suppose that 
            $f \in 
            \AC^0_3[2^{o(n)}]$.
            Let $\alpha : \bbn \to \bbn$ be such that
            $\alpha(n) \to_n \infty$
            and
            $f$ has a depth-3 $\AC^0$ circuit
            of size
            $2^{n/\alpha}$.
            Let $m = 2^{n/(10\cdot\alpha)}$
            and
            let $g : \blt^n \times \blt^m \to \blt$
            be the function
            $g(x,y) = f(x)$.
            Let $N = n+m = (1+o(1))2^{n/(10\cdot\alpha)}$.
            Clearly, the function $g$ has 
            a depth-3 $\AC^0$ circuit 
            of size $2^{n/\alpha} = N^{O(1)}$.
            Since $g$ is monotone,
            we conclude from the assumption 
            that $g$ is computed by a polynomial-size monotone formula.
            Now, since $f(x) = g(x,1^{m})$,
            we obtain a monotone formula of size
            $N^{O(1)} = 2^{o(n)}$ for computing $f$,
             which contradicts the lower bound of
             Theorem~\ref{thm:pr17-fnp}.
        \item 
            Similar to the proof of item~(\ref{item:ac03}),
            but using \Cref{thm:ckr20} instead.
        \item 
            Suppose that $\NC^2 \sseq \calc$.
            By \Cref{thm:xor-sat}, there exists a monotone function $f \in \NC^2$
            on $n$ bits
            and a number $\eps > 0$
            such that
            $f \notin \mSIZE[2^{o(n^{\eps})}]$.
            Therefore, for any $\delta > 0$ such that $\delta < \eps$,
            we have $f \notin \mSIZE[2^{O(n^\delta)}]$.
            Since, by assumption, we have $f \in \NC^2 \sseq \calc$,
            we obtain
            $\semMon{\calc} \not\sseq \mSIZE[2^{O(n^\delta)}]$.
        \item 
            If $\NC^2 \sseq \NC^1$, then,
            by 
            item (\ref{item:subexp}),
            we get
            $\semMon{\NC^1} \not\sseq \mSIZE[2^{o(n^\eps)}]$.
            From \Cref{thm:nc1-ac0-size},
            we obtain
            $\semMon{\AC^0} 
            \not\sseq \mSIZE[\poly]$.
            \qedhere
    \end{enumerate}
\end{proof}

As a motivation to 
the ensuing discussion,
we recall a result of
Rossman~\cite{rossman_2008},
who 
showed that any homomorphism-preserving
graph property computed by \(\AC^0\) circuits is also computed
by monotone \(\AC^0\) circuits.

\begin{theorem}[\cite{rossman_2008}]
    \label{t:ros08}
    $\AC^0
    \cap 
    \HomPres 
    \sseq \mDNF[\poly]$.
\end{theorem}

This inspires the question of whether general graph properties
can also be efficiently simulated by monotone circuits.
We show that, if true,
such simulations would imply strong complexity separations.
Let us first recall an exponential monotone circuit lower bound for
monotone graph properties, and we will be ready to state and prove our main
result.

\begin{theorem}[\cite{DBLP:journals/combinatorica/AlonB87}]
    \label{thm:ab-graph}
    There exists $\eps > 0$
    such that
    $\semMon{\NP} \cap \GraphPpts \not\sseq \mSIZE[2^{o(n^\eps)}]$.
\end{theorem}

\begin{theorem}
    \label{cor:matching-ag}
    Let $\calc$ be 
    a class of circuits.
    The following holds:
    \begin{enumerate}
        \item 
            \label{item:gpsize}
            If $\semMon{\calc} \cap \GraphPpts \sseq \mSIZE[\poly]$,
            then 
            $\L
            \not\sseq
            \calc$.
        \item 
            \label{item:gpdepth}
            If 
            $\semMon{\calc} \cap \GraphPpts \sseq 
            \mDEPTH[o(\sqrt{n})]$, where $n$ denotes the number of input bits,
            then
            $\L
            \not\sseq
            \calc$.
        \item 
            \label{item:kclq}
            If 
            $\semMon{\AC^0} \cap \GraphPpts \subseteq \mSIZE[\poly]$,
            then $\NP \not\sseq \NC^1$.
    \end{enumerate}
\end{theorem}
\begin{proof}
    We will prove each item separately.

    \begin{enumerate}[label=\emph{Proof of (\arabic*)}., leftmargin=*]
        \item 
            Suppose that $\L \sseq \calc$.
            By \Cref{thm:bip-oddfactor-lb},
            the monotone graph property 
            $\OddFactor$
            satisfies
            $\OddFactor \notin \mSIZE[\poly]$.
            Moreover, we have the upper bound $\OddFactor \in \L$ by
            \Cref{thm:oddfactor-logspace}.
            Since, by assumption, we have $\OddFactor \in \L \sseq \calc$,
            we obtain
            $\semMon{\calc} \cap \GraphPpts \not\sseq \mSIZE[\poly]$.
        \item 
            Suppose that $\L \sseq \calc$.
            By 
            \Cref{thm:oddfactor-logspace,thm:bip-oddfactor-lb}, 
            there exists a monotone graph property $f \in \L$ such
            that $f \notin \mDEPTH[o(\sqrt{n})]$.
            Since, by assumption, we have
            $f \in \L \sseq \calc$,
            we obtain
            $\semMon{\calc} \cap \GraphPpts \not\sseq \mDEPTH[o(\sqrt{n})]$.

        \item 
            Suppose that $\NP \sseq \NC^1$.
            By \Cref{thm:ab-graph}, there exists a monotone graph property
            $f \in \NC^1$ such that
            $\Cmon{f} = 2^{\Omega(n^\eps)}$ for some $\eps > 0$.
            Let $\delta = \eps/2$.
            By Lemma~\ref{lemma:padding_graph} (Item $2$),
            there exists a monotone graph property $g$ on
            $N = 2^{n^\delta}$ vertices
            computed by
            an $\AC^0$ circuit of size $N^{2+o(1)}$
            such that $f$ is a monotone projection of $g$.
            \Cref{thm:ab-graph}
            implies that 
            any monotone circuit computing $f$ has size
            $2^{\Omega(n^{\eps})} 
            = 2^{\Omega((\log N)^2)} = N^{\omega(1)}$.
            \qedhere
    \end{enumerate}
\end{proof}

\section{Monotone Complexity of Constraint Satisfaction Problems}
\label{sec:CSPs}

In this section, we study the monotone complexity of Boolean-valued $\CSP$s.
Our goal is to classify which types of Boolean $\CSP$s are hard for
monotone circuit size and monotone circuit depth,
eventually proving
\Cref{thm:cspsat-intro,thm:intro-mon-dichotomies}.

We will first spend some time recalling standard definitions and concepts
in the theory of CSPs (\Cref{sec:csp-def}),
as well as a few results about CSPs that were
proved in previous works
\cite{DBLP:conf/stoc/Schaefer78,DBLP:journals/tcs/Jeavons98,
playing_one,playing_two,DBLP:journals/jcss/AllenderBISV09}
(\Cref{sec:csp-fact}).
We will then prove \Cref{thm:intro-mon-dichotomies}
in \Cref{sec:csp-dichotomy},
and
we will finally prove \Cref{thm:cspsat-intro}
in \Cref{sec:csp-conseq}
after 
proving some auxiliary results
in \Cref{sec:csp-ac0}.

\subsection{Definitions}
\label{sec:csp-def}

For a good introduction to the concepts defined below, we refer the reader
to~\cite{playing_one, playing_two}.
We also refer the reader to 
\Cref{sec:intro_mon_csp} 
for 
the definition of the family of functions
$\CSPf{S}$, as well as the terms
\emph{constraint application, $S$-formula} and \emph{satisfiable formula}.

We denote by $p_i^n : \blt^n \to \blt$ the
$i$-th \emph{projection function} on $n$ variables,
whose operation is defined as
$p_i^n(x) = x_i$.
For a set of Boolean functions $B$,
we
denote by $[B]$ the \emph{closure} of $B$,
defined as follows:
a Boolean function $f$ is in $[B]$
if and only if
$f \in B \cup \set{\clonefont{Identity}}$
or if there exists $g \in B$
and $h_1,\dots,h_k$
such that
$f = g(h_1,\dots,h_k)$,
where each $h_i$ is either a 
projection function or a function from $[B]$.
We can equivalently define $[B]$ as
the set of all Boolean functions
that can be computed by circuits using 
the functions of $B$ as gates.
Note that $[B]$ necessarily contains an infinite number of Boolean
functions, 
since $p_1^n \in [B]$ for every $n \in \bbn$;
moreover, the constant functions
are not necessarily in $[B]$.
We say that $B$ is a \emph{clone} if $B = [B]$.
A few prominent examples of clones are the set of all Boolean functions
(equal to $[\set{\land,\neg}]$),
monotone functions
(equal to $[\set{\land,\lor,0,1}]$), 
and linear functions
(equal to $[\set{\xor,1}]$).

\begin{remark}
    \label{rk:post_lattice}
    The set of all clones forms a lattice, known as
    \defx{Post's lattice},
    under the operations
    $[A] \sqcap [B] := [A] \cap [B]$
    and
    $[A] \sqcup [B] := [A \cup B]$.
    From the next section onwards,
    we will refer to the clones defined in~\cite{playing_one}
    (such as $\clonefont{I_0}$, $\clonefont{I_1}$, etc.),
    assuming the reader is familiar with them.
    For the unfamiliar reader, we refer to \Cref{sec:clone-background} and 
    \Cref{fig:clone_table,fig:post_lattice},
    which contain all the definitions of the clones we will need,
    as well as the entire Post's lattice in graphical representation.
    
    To avoid confusion, we will always refer to clones with
    $\clonefont{normal}$-$\clonefont{Roman}$ font (e.g.,
    $\clonefont{S_1}, \clonefont{I_0}$, etc).
\end{remark}

Let $S$ be a finite set of Boolean relations.
We denote by $\CNF(S)$ the set of all $S$-formulas.
We denote by $\COQ(S)$ 
the set of all relations which can be expressed
with the following type of formula $\varphi$:
\begin{equation*}
    \varphi(x_1,\dots,x_k) = \exists y_1,\dots,y_\ell
    \,
    \psi(x_1,\dots,x_k,y_1,\dots,y_\ell),
\end{equation*}
where $\psi \in \CNF(S)$.
The relations in $\COQ(S)$ will also be referred as \emph{conjunctive
queries} over $S$.
We denote by $\ccln{S}$ the set of relations
defined as
$\ccln{S} := \COQ(S \cup \set{=})$.
If $S = \ccln{S}$, we say that $S$ is a \emph{co-clone}.
We define 
$$\CSP = \set{\CSPf{S} : \text{$S$ is a finite set of relations}}.$$
We say that $\CSPf{S}$ is \emph{trivial}
if $\CSPf{S}$ is a constant function.

Let $R$ be a $k$-ary Boolean relation 
and
let $f : \blt^\ell \to \blt$ be a Boolean function.
For $x \in R$ and $i \in [k]$, we denote by $x[i]$
the $i$-th bit of $x$.
\begin{definition}
    \label{def:polymorphism}
    We say that $f$ is a
    \defx{polymorphism} of $R$,
    and $R$ is 
    \defx{an invariant of $f$},
    if, for all $x_1, \dots, x_\ell \in R$, we have
    \begin{equation*}
        (
            f(x_1[1],\dots,x_{\ell}[1]),
            f(x_2[2],\dots,x_{\ell}[2]),
            \dots,
            f(x_k[k],\dots,x_{\ell}[k])
        )
        \in
        R.
    \end{equation*}
    We denote the set of all polymorphisms of $R$ by $\Pol(R)$.
    For a set of relations $S$,
    we denote by $\Pol(S)$ the set of Boolean functions
    which are polymorphisms of all the relations of $S$.
    For a set of Boolean functions, we denote by
    $\Inv(B)$ the set of all Boolean relations which are invariant
    under all functions of $B$ (i.e.,
    $\Inv(B) = \set{R : B \sseq \Pol(R)}$).
\end{definition}

The following summarises the important facts about
clones, co-clones and polymorphisms that are relevant to
the study of CSPs~\cite{10.1145/263867.263489}. 

\begin{lemma}
    \label{lem:galois}
    Let 
    $S$ 
    and $S'$
    be sets
    of Boolean relations 
    and
    let
    $B$ 
    and $B'$
    be sets
    of Boolean functions.
    We have
    \begin{enumerate}[label=(\roman*)]
        \item 
            \label{item:pol-clone}
            $\Pol(S)$ is a clone and $\Inv(B)$ is a co-clone;
        \item 
            \label{item:pol-invert}
            If 
            $S \sseq S'$,
            then
            $\Pol(S') \sseq \Pol(S)$;
        \item 
            \label{item:inv-invert}
            If 
            $B \sseq B'$,
            then
            $\Inv(B') \sseq \Inv(B)$;
        \item 
            \label{item:coq-idempotent}
            $\COQ(\COQ(S)) = \COQ(S)$;
        \item 
            \label{item:coq-sseq}
            If $S \sseq S'$, then
            $\COQ(S) \sseq \COQ(S')$;
        \item 
            \label{item:invpol}
            $\Inv(\Pol(S)) = \ccln{S}$;
        \item 
            \label{item:polinv}
            $\Pol(\Inv(B)) = [B]$.
    \end{enumerate}
\end{lemma}

We now define 
different
types of reductions.
We say that a reduction is a \emph{monotone OR-reduction}
if
every bit of the reduction is either constant
or can be computed by a monotone disjunction on the
input variables.
We write $f \mormred g$ if 
there exists a many-one monotone OR-reduction
from $f$ to $g$.
We also write 
$f \acmred g$
if 
there exists a many-one
$\AC^0$ reduction
from $f$ to $g$,
and
$f \mnlmred g$
if 
there exists a many-one
$\mNL$ reduction
from $f$ to $g$\footnote{A many-one $\AC^0$ (resp. $\mNL$) reduction is one
in which each bit of the reduction is either constant or 
can be computed with a polynomial-size $\AC^0$ circuit
(resp. monotone nondeterministic branching program). 
Recall that a 
\emph{monotone nondeterministic branching program} 
is a directed acyclic graph $G$ with two
distinguished vertices $s$ and $t$, 
in which each edge $e$ is labelled with an input function 
$\rho_e \in \set{1,x_1,\dots,x_n}$. 
Given an input $x$, the program accepts if there
exists a path from $s$ to $t$ in the subgraph $G_x$ of $G$ 
in which an edge $e$ appears if $\rho_e(x)=1$.
}.
Unless otherwise specified, every reduction we consider will generate an instance
of polynomial size on the length of the input.

Finally, we denote by $\OR^k$ and $\NAND^k$
the $k$-ary $\OR$ and $\NAND$ relations, respectively.

\subsection{Basic facts about $\CSPSAT$}
\label{sec:csp-fact}

We state 
here basic facts about the $\CSPSAT$ function.
These facts are proved in the original paper of
Schaefer~\cite{DBLP:conf/stoc/Schaefer78},
as well as in later
papers~\cite{DBLP:journals/tcs/Jeavons98,
playing_one,playing_two,DBLP:journals/jcss/AllenderBISV09}.

\Cref{prop:poly} below is one of the most important lemmas of this section and
will be used many times. It states that $\Pol(S)$
characterises the monotone complexity of $\CSPf{S}$, in the sense that 
the sets of relations with few polymorphisms give rise to the hardest
instances of CSPs.
A non-monotone version of this result was proved in
\cite[Theorem 2.4]{DBLP:journals/tcs/Jeavons98, playing_two},
and we
check 
in \Cref{sec:reduction_monotone}
that their proofs also hold in the monotone
case.

\begin{restatable}[\protect{Polymorphisms characterise the complexity of
    CSPs~\cite[Theorem 2.4]{DBLP:journals/tcs/Jeavons98, playing_two}}]
    {lemma}{polym}
    \label{prop:poly}
    If $\Pol(S_2) \sseq \Pol(S_1)$,
    then
    $\CSPfn{S_1}{n} \mnlmred \CSPfn{S_2}{\poly(n)}$.
\end{restatable}

\Cref{thm:mon-schaefer} gives monotone circuit upper bounds
for some instances of $\CSPf{S}$.
Non-monotone variants of this upper bound were originally obtained
in the seminal paper of Schaefer~\cite{DBLP:conf/stoc/Schaefer78},
and we again check that the monotone variants work in
\Cref{sec:monotone_upper}.

\begin{restatable}[Monotone version of the upper bounds 
    for $\CSPSAT$ 
\cite{DBLP:conf/stoc/Schaefer78,DBLP:journals/jcss/AllenderBISV09}]
{theorem}
{monschaeffer}
    \label{thm:mon-schaefer}
    Let $S$ be a finite set of relations.
    The following holds.
    \begin{enumerate}
        \item
            If $\clonefont{\clonefont{E_2}} \sseq \Pol(S)$ or $\clonefont{\clonefont{V_2}} \sseq \Pol(S)$,
            then
            $\CSPf{S} \in \mSIZE[\poly]$.
        \item 
            If 
            $\clonefont{\clonefont{D_2}} \sseq \Pol(S)$,
            or
            $\clonefont{\clonefont{S_{00}}} \sseq \Pol(S)$,
            or
            $\clonefont{\clonefont{S_{10}}} \sseq \Pol(S)$,
            then
            $\CSPf{S} \in \mNL$.
    \end{enumerate}
\end{restatable}

Finally, 
we state here a result of~\cite{DBLP:journals/jcss/AllenderBISV09},
which
classifies the \emph{non-monotone} complexity of $\CSPf{S}$
under $\acmred$ reductions.
The classification of the complexity of $\CSPf{S}$
is based solely on $\Pol(S)$.
See \Cref{fig:post_lattice} for a graphical representation.

\begin{theorem}[\protect{Refined classification of $\CSP$ problems 
    \cite[Theorem 3.1]{DBLP:journals/jcss/AllenderBISV09}}]
    \label{thm:refined-completeness}
    Let $S$ be a finite set of Boolean relations.
    The following holds.
    \begin{itemize}
        \item If 
            $\clonefont{I_0} \sseq \Pol(S)$
            or 
            $\clonefont{I_1} \sseq \Pol(S)$,
            then
            $\CSPf{S}$ is trivial.
        \item 
            If 
            $\Pol(S) \in \set{\clonefont{I_2},\clonefont{N_2}}$,
            then
            $\CSPf{S}$
            is
            \acmcmpl{}
            for
            $\NP$.
        \item 
            If 
            $\Pol(S) \in 
            \set{\clonefont{V_2},\clonefont{E_2}}$,
            then
            $\CSPf{S}$
            is
            \acmcmpl{}
            for
            $\P$.
        \item 
            If 
            $\Pol(S) \in 
            \set{\clonefont{L_2},\clonefont{L_3}}$,
            then
            $\CSPf{S}$
            is
            \acmcmpl{}
            for
            $\pL$.
        \item 
            If 
            $\clonefont{S_{00}} \sseq \Pol(S) \sseq \clonefont{S_{00}}^2$
            or
            $\clonefont{S_{10}} \sseq \Pol(S) \sseq \clonefont{S_{10}}^2$
            or
            $\Pol(S) \in 
            \set{\clonefont{D_2},\clonefont{M_2}}$,
            then
            $\CSPf{S}$
            is
            \acmcmpl{}
            for
            $\NL$.
        \item 
            If 
            $\Pol(S) \in 
            \set{\clonefont{D_1}, \clonefont{D}}$,
            then
            $\CSPf{S}$
            is
            \acmcmpl{}
            for
            $\L$.
        \item 
            If 
            $\clonefont{S_{02}} \sseq \Pol(S) \sseq \clonefont{R_2}$
            or
            $\clonefont{S_{12}} \sseq \Pol(S) \sseq \clonefont{R_2}$,
            then
            either
            $\CSPf{S} \in \AC^0$
            or
            $\CSPf{S}$
            is
            \acmcmpl{}
            for
            $\L$.
    \end{itemize}
\end{theorem}

\newcommand\ABISV
    {R2/.cstyle=LN,
     M/.cstyle=NA,
     M2/.cstyle=NL,
     D/.cstyle=L,
     D2/.cstyle=NL,
     S212/.cstyle=LN,
     S312/.cstyle=LN,
     S12/.cstyle=LN,
     S202/.cstyle=LN,
     S302/.cstyle=LN,
     S02/.cstyle=LN,
     S00/.cstyle=NL,
     S300/.cstyle=NL,
     S200/.cstyle=NL,
     S10/.cstyle=NL,
     S310/.cstyle=NL,
     S210/.cstyle=NL,
     E2/.cstyle=P,
     V2/.cstyle=P,
     L2/.cstyle=pL,
     L3/.cstyle=pL,
     N2/.cstyle=NP,
     I2/.cstyle=NP,
     xscale=-1, yscale=-1,scale=1,labels=wagner,cshape/circle/.append style={minimum width=15pt}
    }
\newlength{\boxwidth}
\settowidth{\boxwidth}{$\NP$-complete}
\addtolength{\boxwidth}{5mm}
\begin{figure}
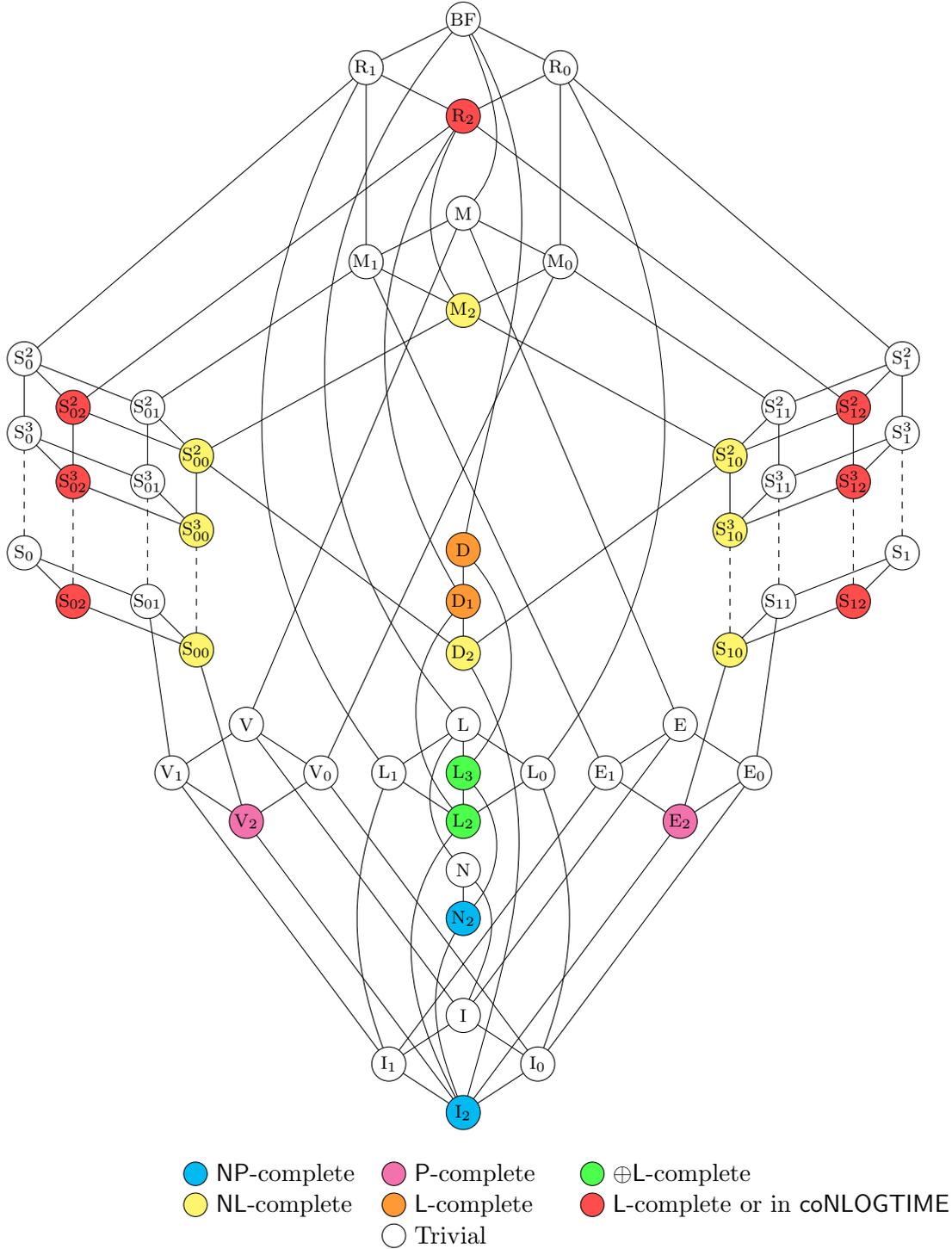

    \centering
    \expandafter\postlattice\expandafter[\ABISV]
    \bigskip

    \begin{tabular}{lll}
        \makebox[\boxwidth][l]{\protect\postlegend{NP} $\NP$-complete}
        & \makebox[\boxwidth][l]{\protect\postlegend{P} $\P$-complete}
        & \makebox[\boxwidth][l]{\protect\postlegend{pL} $\pL$-complete}
        \\
        \makebox[\boxwidth][l]{\protect\postlegend{NL} $\NL$-complete}
        & \makebox[\boxwidth][l]{\protect\postlegend{L} $\L$-complete}
        & \makebox[\boxwidth][l]{\protect\postlegend{LN} $\L$-complete or in  $\mathsf{coNLOGTIME}$}
        \\
        & \makebox[\boxwidth][l]{\protect\postlegend{NA} Trivial}
        &
        \\
    \end{tabular}

    \caption{
        Graph of all closed classes of Boolean functions.
        The vertices are colored with the complexity of 
        deciding $\CSP$s 
        whose set of polymorphisms corresponds to the label of the vertex.
        Trivial $\CSP$s are those that correspond to constant functions.
        Every hardness result is proved under $\acmred$ reductions.
        See \Cref{thm:refined-completeness} for details.
        A similar figure appears in
        \cite[Figure 1]{DBLP:journals/jcss/AllenderBISV09}.}
    \label{fig:post_lattice}
\end{figure}

\subsection{A monotone dichotomy for $\CSPSAT$}
\label{sec:csp-dichotomy}

In this section, we prove \Cref{thm:intro-mon-dichotomies}.
We first prove Part (1) of the theorem (the dichotomy for circuit size),
and then we prove Part (2) of the theorem (the dichotomy for circuit
depth).

\subparagraph{Dichotomy for circuits.}

To prove the dichotomy for circuits, 
we first show that, for
any set of relations $S$ whose set of polymorphisms is contained in $\clonefont{L_3}$,
we can monotonically reduce $\txorsat$ to $\CSPf{S}$.

\begin{lemma}
    \label{lem:l3_xorsat}
    Let $S$ be a finite set of relations.
    If 
    $\Pol(S) \sseq \clonefont{L_3}$,
    then
    $\txorsat \mnlmred \CSPf{S}$.
\end{lemma}
\begin{proof}
    Inspecting Post's lattice
    (\Cref{fig:post_lattice}),
    note that the only clones 
    strictly contained in
    $\clonefont{L_3}$ are
    $\clonefont{L_2}, \clonefont{N_2}$ and $\clonefont{I_2}$.
    We will first show that the reduction holds for the case
    $\Pol(S)=\clonefont{L_2}$ and then prove that 
    the reduction also holds for the case $\Pol(S)=\clonefont{L_3}$.
    \Cref{prop:poly}
    will then imply the cases $\Pol(S) \in \set{\clonefont{N_2},\clonefont{I_2}}$, since
    $\clonefont{I_2} \sseq \clonefont{N_2} \sseq \clonefont{L_3}$.

    It's not hard to check that,
    if $\Pol(S) = \clonefont{L_2}$,
    then $\Pol(S) \sseq \Pol(\txorsat)$
    (it suffices to observe that bitwise XORing
    three satisfying assignments to a linear equation gives rise to 
    a new satisfying assignment to the same equation).
    Therefore, from \Cref{prop:poly} we deduce
    that $\txorsat$ admits a reduction to
    $\CSPf{S}$ in $\mNL$.
    In order to prove the case $\Pol(S) = \clonefont{L_3}$, we first prove the
    following claim.

    \begin{claim*}[\protect{\cite[Lemma 3.11]{DBLP:journals/jcss/AllenderBISV09}}]
        \label{lem:l2-to-l3}
        Let $S$ be a finite set of relations such that
        $\Pol(S) = \clonefont{L_2}$.
        There exists a finite set of relations $S'$
        such that
        $\Pol(S') = \clonefont{L_3}$
        and
        $\CSPfn{S}{n} \mprojred \CSPfn{S'}{n+1}$.
    \end{claim*}
    \begin{proof}
        We describe the proof of Lemma 3.11
        in~\cite{DBLP:journals/jcss/AllenderBISV09}
        and observe that it gives a monotone reduction.

        For a relation $R \in S$,
        let $R' = \set{(\neg x_1,\dots, \neg x_k) : (x_1,\dots,x_k) \in R}$.
        Let also $S' = \set{R' : R \in S}$.
        It's not hard to check that 
        $\Pol(S') = \clonefont{L_3}$,
        since $S'$ is 
        an invariant of
        $\clonefont{L_2}$ and $\clonefont{N_2}$,
        and $\clonefont{L_3}$
        is the smallest clone containing both $\clonefont{L_2}$ and $\clonefont{N_2}$;
        moreover, if $\rho \in \Pol(S')$ and $\rho$ is a Boolean function on at
        least two bits, then
        $\rho \in \Pol(S) = \clonefont{L_2}$.

        Now let 
        $F$
        be a instance of $\CSPfn{S}{n}$.
        For every constraint $C = R(x_1,\dots,x_k)$ in $F$,
        we
        add the constraint
        $C' = R'(\alpha,x_1,\dots,x_k)$ to the $S'$-formula $F'$,
        where $\alpha$ is a new variable.
        Note that $F'$ is a $S'$-formula, defined on $n+1$ variables,
        which
        is satisfiable 
        if and only if $F$ is satisfiable.
        Moreover, the construction of $F'$ from $F$ can be done with a monotone
        projection.
    \end{proof}

    Since the case $\Pol(S) = \clonefont{L_2}$ holds,
    the case $\Pol(S) = \clonefont{L_3}$ now follows
    from~\Cref{prop:poly} and the Claim.
    Finally, from \Cref{prop:poly}
    we conclude that the 
    reduction
    also holds for the case
    $\Pol(S) \in \set{\clonefont{N_2},\clonefont{I_2}}$, since
    $\clonefont{I_2} \sseq \clonefont{N_2} \sseq \clonefont{L_3}$.
\end{proof}

\begin{theorem}[Dichotomy for monotone circuits]
    \label{thm:monckt-dichotomy}
    Let $S$ be a finite set of relations.
    If 
    $\Pol(S) \sseq \clonefont{L_3}$
    then there exists a constant $\eps > 0$
    such that
    $\Cmon{\CSPf{S}} = 2^{\Omega(n^\eps)}$.
    Otherwise, 
    we have
    $\Cmon{\CSPf{S}} = n^{O(1)}$.
\end{theorem}
\begin{proof}
    If $\Pol(S) \sseq \clonefont{L_3}$, the lower bound follows from
    the 'moreover' part of 
    \Cref{thm:xor-sat}, and 
    \Cref{lem:l3_xorsat}.
    For the upper bound, we inspect Post's
    lattice~(\Cref{fig:post_lattice}).
    Observe that, if $\Pol(S) \not\sseq \clonefont{L_3}$, the following are the only possible cases:
    \begin{enumerate}
        \item $\clonefont{I_0} \sseq \Pol(S)$ or $\clonefont{I_1} \sseq \Pol(S)$.
            In both cases, any $\CNF(S)$ is trivially satisfiable.
        \item $\clonefont{E_2} \sseq \Pol(S)$ or $\clonefont{V_2} \sseq \Pol(S)$.
            In this case, $\CSPf{S} \in \mSIZE[\poly]$ 
            by~\Cref{thm:mon-schaefer}.
        \item $\clonefont{D_2} \sseq \Pol(S)$.
            In this case, $\CSPf{S} \in \mNL \sseq \mSIZE[\poly]$ 
            by~\Cref{thm:mon-schaefer}.
            \qedhere
    \end{enumerate}
\end{proof}

\tikzset{NA/.style  = {fill = white}}
\tikzset{mP/.style  = {fill = magenta!50}}
\tikzset{LB/.style  = {pattern = north east lines, pattern color = applegreen!80, line width=1.5pt}}
\newcommand\CktDichotomy
    {R2/.cstyle=mP,
     M/.cstyle=NA,
     M2/.cstyle=mP,
     D/.cstyle=mP,
     D2/.cstyle=mP,
     S212/.cstyle=mP,
     S312/.cstyle=mP,
     S12/.cstyle=mP,
     S202/.cstyle=mP,
     S302/.cstyle=mP,
     S02/.cstyle=mP,
     S00/.cstyle=mP,
     S300/.cstyle=mP,
     S200/.cstyle=mP,
     S10/.cstyle=mP,
     S310/.cstyle=mP,
     S210/.cstyle=mP,
     E2/.cstyle=mP,
     V2/.cstyle=mP,
     L2/.cstyle=LB,
     L3/.cstyle=LB,
     N2/.cstyle=LB,
     I2/.cstyle=LB,
     xscale=-1, yscale=-1,scale=1,labels=wagner,cshape/circle/.append style={minimum width=15pt}
    }
\newlength{\cktdichotomybox}
\settowidth{\cktdichotomybox}{Requires monotone circuits of large size}
\addtolength{\cktdichotomybox}{5mm}
\begin{figure}
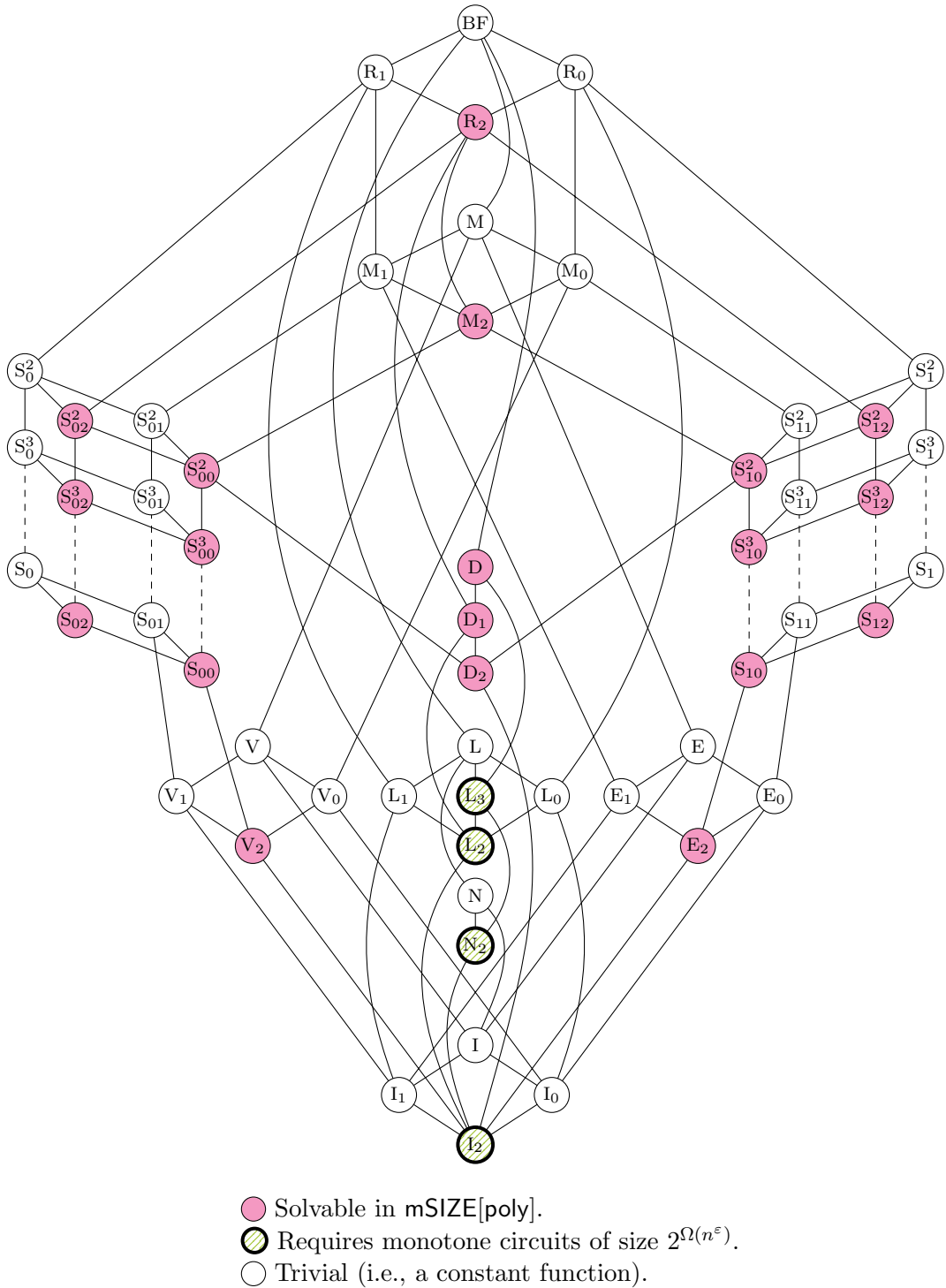

    \centering
    \expandafter\postlattice\expandafter[\CktDichotomy]
    \bigskip

    \begin{tabular}{l}
        \makebox[\cktdichotomybox][l]{\protect\postlegend{mP} Solvable in
        $\mSIZE[\poly]$.}
        \\
        \makebox[\cktdichotomybox][l]{\protect\postlegend{LB} Requires
        monotone circuits of size $2^{\Omega(n^{\eps})}$.}
        \\
        \makebox[\cktdichotomybox][l]{\protect\postlegend{NA} 
        Trivial (i.e., a constant function).}
    \end{tabular}

    \caption{Illustration of \Cref{thm:monckt-dichotomy}.
    The vertices are colored with the 
    \emph{monotone circuit size} complexity of deciding CSPs whose set of
    polymorphisms corresponds to the label of the vertex.
    }
    \label{fig:monckt-dichotomy}
\end{figure}

\begin{remark}
    \label{rk:lifting-approximation}
    We remark that the lifting theorem of~\cite{gkrs19}
    (which is an ingredient in the proof of \Cref{thm:xor-sat})
    is only used to prove that the monotone complexity of 
    $\CSPf{S}$
    is exponential
    when
    $\Pol(S) \sseq \clonefont{L_3}$.
    If we only care to show a superpolynomial separation, then 
    it suffices to apply the superpolynomial lower bound for CSPs 
    with counting proved
    in~\cite{DBLP:journals/siamcomp/FederV98,DBLP:journals/combinatorica/BabaiGW99}
    using the approximation method.
    Indeed, we give an explicit proof in~\Cref{sec:xorsat-appr}.
    The same holds for the consequences of this theorem
    (see \Cref{thm:cspsat-hard}).
\end{remark}

\subparagraph{Dichotomy for formulas.}

Define $\thornsat_n : \blt^{2n^3+n} \to \blt$ 
as $\thornsat_n = \CSPfn{\hornt}{n}$, where
$$
\hornt = \set{
    (\neg x_1 \lor \neg x_2 \lor x_3),
    (\neg x_1 \lor \neg x_2 \lor \neg x_3),
    (x)
}.$$
The following is proved 
in~\cite{DBLP:journals/combinatorica/RazM99, gkrs19}.

\begin{theorem}[\cite{DBLP:journals/combinatorica/RazM99, gkrs19}]
    \label{thm:3hornsat_lb}
    There exists $\eps > 0$
    such that
    $\thornsat \in \mSIZE[\poly] \sm \mDEPTH[o(n^\eps)]$.
\end{theorem}
\begin{proof}[Proof sketch]
    Since $\clonefont{E_2} \sseq \Pol(\hornt)$ (see,
    e.g.,~\cite[Lemma 4.8]{DBLP:books/daglib/0004131}),
    the upper bound follows from~\Cref{thm:mon-schaefer}.
    The lower bound follows from a lifting theorem
    of~\cite{DBLP:journals/combinatorica/RazM99,gkrs19}.
    They show that 
    the monotone circuit-depth of $\thornsat$
    is at least the
    depth of the smallest Resolution-tree
    refuting a so-called \emph{pebbling formula}.
    Since this formula requires Resolution-trees of depth $n^\eps$,
    the lower bound follows.
\end{proof}

Analogously to the previous section, we show that $\thornsat$ reduces
to $\CSPf{S}$ whenever
$\Pol(S)$ is small enough, in 
a precise sense stated below.
We then deduce the dichotomy for formulas with a similar argument.

\begin{lemma}
    \label{lem:l3_hornsat}
    Let $S$ be a finite set of relations.
    If 
    $\Pol(S) \sseq \clonefont{E_2}$ or $\Pol(S) \sseq \clonefont{V_2}$,
    then
    $\thornsat \mnlmred \CSPf{S}$.
\end{lemma}
\begin{proof}
    We first consider the case $\Pol(S) \sseq \clonefont{E_2}$.
    Note that
    $\clonefont{E_2} \sseq \Pol(\thornsat)$
    (see, e.g., \cite[Lemma 4.8]{DBLP:books/daglib/0004131}).
    Therefore, from \Cref{prop:poly} we deduce
    that $\thornsat$ admits a reduction to
    $\CSPf{S}$ in $\mNL$.

    Now let
    $\tahornsat = \CSPf{\ahornt}$,
    where
    $\ahornt = 
    \set{
        (x_1 \lor x_2 \lor \neg x_3),
        (x_1 \lor x_2 \lor x_3),
        (\neg x)
    }$.
    Observe that a
    $\hornt$-formula $\varphi$
    is satisfiable if and only if
    the $\ahornt$-formula 
    $\varphi(\neg x_1,\dots,\neg x_n)$
    is satisfiable.
    Therefore, we have
    $\thornsat \mprojred \tahornsat$.
    Observing that
    $\clonefont{V_2} \sseq \Pol(\ahornt)$
    (again, see e.g. \cite[Lemma 4.8]{DBLP:books/daglib/0004131}),
    the result now follows from
    \Cref{prop:poly} and the previous paragraph.
\end{proof}

\begin{theorem}[Dichotomy for monotone formulas]
    \label{thm:monfml-dichotomy}
    Let $S$ be a finite set of relations.
    If 
    $\Pol(S) \sseq \clonefont{L_3}$, or
    $\Pol(S) \sseq \clonefont{V_2}$, or
    $\Pol(S) \sseq \clonefont{E_2}$,
    then there 
    is
    a constant $\eps > 0$
    such that
    $\mdepth{\CSPf{S}} = {\Omega(n^\eps)}$.
    Otherwise, 
    we have
    $\CSPf{S} \in \mNL
    \sseq \mNC^2 \sseq \mDEPTH[\log^2 n]$.
\end{theorem}
\begin{proof}
    We will first prove the lower bound.
    If $\Pol(S) \sseq \clonefont{L_3}$, the lower bound follows
    from \Cref{thm:monckt-dichotomy}.
    If
    $\Pol(S) \sseq \clonefont{V_2}$ or
    $\Pol(S) \sseq \clonefont{E_2}$,
    the lower bound follows 
    from \Cref{thm:3hornsat_lb,lem:l3_hornsat}.
    
    By inspecting Post's 
    lattice~(\Cref{rk:post_lattice}),
    we see that
    the remaning cases are:
    \begin{enumerate}
        \item $\clonefont{I_0} \sseq \Pol(S)$ or $\clonefont{I_1} \sseq \Pol(S)$.
            In both cases, any $\CNF(S)$ is trivially satisfiable.
        \item 
            $\clonefont{S_{00}} \sseq \Pol(S)$,
            or
            $\clonefont{S_{10}} \sseq \Pol(S)$,
            or
            $\clonefont{D_2} \sseq \Pol(S)$.
            In all of those three cases, we have
            $\CSPf{S} \in \mNL$ by
            \Cref{thm:mon-schaefer}.
            \qedhere
    \end{enumerate}
\end{proof}

\tikzset{mD/.style = {fill = mikadoyellow}}
\tikzset{dLB/.style = {pattern = north east lines, pattern color = cyan(process)!70, line width=1.5pt}}
\newcommand\FmlDichotomy
    {R2/.cstyle=mD,
     M/.cstyle=NA,
     M2/.cstyle=mD,
     D/.cstyle=mD,
     D2/.cstyle=mD,
     S212/.cstyle=mD,
     S312/.cstyle=mD,
     S12/.cstyle=mD,
     S202/.cstyle=mD,
     S302/.cstyle=mD,
     S02/.cstyle=mD,
     S00/.cstyle=mD,
     S300/.cstyle=mD,
     S200/.cstyle=mD,
     S10/.cstyle=mD,
     S310/.cstyle=mD,
     S210/.cstyle=mD,
     E2/.cstyle=dLB,
     V2/.cstyle=dLB,
     L2/.cstyle=dLB,
     L3/.cstyle=dLB,
     N2/.cstyle=dLB,
     I2/.cstyle=dLB,
     xscale=-1, yscale=-1,scale=1,labels=wagner,cshape/circle/.append style={minimum width=15pt}
    }
\begin{figure}
    \centering
    \expandafter\postlattice\expandafter[\FmlDichotomy]
    \bigskip

    \begin{tabular}{l}
        \makebox[\cktdichotomybox][l]{\protect\postlegend{mD} Solvable in
        $\mNL \sseq \mDEPTH[O(\log^2 n)]$.}
        \\
        \makebox[\cktdichotomybox][l]{\protect\postlegend{dLB} Requires
        monotone depth ${\Omega(n^{\eps})}$.}
        \\
        \makebox[\cktdichotomybox][l]{\protect\postlegend{NA} 
        Trivial
        (i.e., a constant function).}
    \end{tabular}

    \caption{Illustration of \Cref{thm:monfml-dichotomy}.
        The vertices are colored with the 
        \emph{monotone circuit depth} complexity of deciding CSPs whose set of
        polymorphisms corresponds to the label of the vertex.
    }
    \label{fig:monfml-dichotomy}
\end{figure}

\subsection{Some auxiliary results}
\label{sec:csp-ac0}

In this section, we prove auxiliary results needed in the proof
of a more general form of \Cref{thm:cspsat-intro}.
In particular, we will prove that all $\CSPf{S}$ which are in $\AC^0$ 
are also contained in $\mAC^0 \sseq \mNC^1$.
Moreover, we show that,
if $\CSPf{S} \notin \mNC^1$,
then $\CSPf{S}$ is $\L$-hard under $\acmred$ reductions.

We first observe that, when $\COQ(S_1) \sseq \COQ(S_2)$,
there exists an efficient low-depth reduction
from $\CSPf{S_1}$ to $\CSPf{S_2}$.
This reduction, which will be useful in this section,
is more refined than the one given by
\Cref{prop:poly}.
A proof of the non-monotone version of this statement
is found in
\cite[Proposition 2.3]{playing_two},
and we give a monotone version of this proof in
\Cref{sec:reduction_monotone}.

\begin{restatable}[\protect{\cite[Proposition 2.3]{playing_two}}]{lemma}{inclcoq}
    \label{prop:incl-coq}
    If $\COQ(S_1) \sseq \COQ(S_2)$,
    then
    there exists a constant $C \in \bbn$ such that
    $\CSPfn{S_1}{n} \mormred \CSPfn{S_2}{C n}$.
\end{restatable}
\begin{proof}
    We defer the proof to
    \Cref{sec:reduction_monotone}.
\end{proof}

We now recall some lemmas
from~\cite{DBLP:journals/jcss/AllenderBISV09},
and prove a few consequences from them.
We say that a set $S$ of relations \emph{can express equality} if
$\set{=} \sseq \COQ(S)$.

\begin{lemma}[\cite{DBLP:journals/jcss/AllenderBISV09}]
    \label{lem:coq-no-eq}
    Let $S$ be a finite set of relations.
    Suppose
    $\clonefont{S_{02}} \sseq \Pol(S)$
    ($\clonefont{S_{12}} \sseq \Pol(S)$, resp.)
    and that $S$ cannot express equality.
    Then there exists $k \geq 2$
    such that
    $S \sseq \COQ(\set{\OR^k, x, \neg x})$
    ($S \sseq \COQ(\set{\NAND^k, x, \neg x})$, resp.).
\end{lemma}
\begin{proof}
    Follows from the proof of Lemma 3.8
    of~\cite{DBLP:journals/jcss/AllenderBISV09}.
\end{proof}

\begin{lemma}
    \label{lem:ac0-cspsat}
    Let $S$ be a finite set of relations
    such that $\Pol(S) \sseq \clonefont{R_2}$.
    If $\clonefont{S_{02}} \sseq \Pol(S)$
    or
    $\clonefont{S_{12}} \sseq \Pol(S)$,
    and $S$ cannot express equality,
    then
    $\CSPf{S} \in 
    \mAC^0_3$.
\end{lemma}
\begin{proof}
    We write the proof in the case
    $\clonefont{S_{02}} \sseq \Pol(S)$.
    The other case is analogous.

    From \Cref{prop:incl-coq,lem:coq-no-eq} and
    \Cref{item:coq-idempotent,item:coq-sseq} of \Cref{lem:galois},
    we get that
    there is a monotone OR-reduction
    from
    $\CSPf{S}$ to
    $\CSPf{\set{\OR^k, x, \neg x}}$ for some $k$.
    However, an $\set{\OR^k, x, \neg x}$-formula
    is unsatisfiable iff
    there exists a literal and its negation as a constraint in the formula,
    or if there exists a disjunction
    in the formula
    such that
    every one of its literals appears negatively as a constraint.
    This condition can be easily checked by a polynomial-size
    monotone DNF.
    Composing the monotone DNF with the monotone OR-reduction,
    we obtain a depth-3 $\AC^0$ circuit computing
    $\CSPf{S}$.
\end{proof}

\begin{lemma}[\protect{\cite[Lemma 3.8]{DBLP:journals/jcss/AllenderBISV09}}]
    \label{lem:equality-lhard}
    Let $S$ be a finite set of relations
    such that $\Pol(S) \sseq \clonefont{R_2}$.
    If 
    $\clonefont{S_{02}} \sseq \Pol(S)$
    or
    $\clonefont{S_{12}} \sseq \Pol(S)$,
    and $S$ can express equality,
    then
    $\CSPf{S}$ is $\L$-hard
    under $\acmred$ reductions.
\end{lemma}

\begin{lemma}
    \label{lem:everything-lhard}
    Let $S$ be a finite set of relations.
    If 
    $\clonefont{S_{02}} \not\sseq \Pol(S)$
    and
    $\clonefont{S_{12}} \not \sseq \Pol(S)$,
    then
    $\CSPf{S}$ is $\L$-hard or trivial.
\end{lemma}
\begin{proof}
    This
    follows by 
    inspecting
    Post's 
    lattice~(\Cref{fig:post_lattice})
    and the classification theorem~(\Cref{thm:refined-completeness}).
\end{proof}

We may now prove the main result of this subsection.

\begin{theorem}
    \label{thm:ac0-csp}
    We have
    $\CSP \cap \AC^0 \sseq \mAC^0_3$.
    Moreover, if $\CSPf{S} \notin 
    \mAC^0_3$,
    then
    $\CSPf{S}$ is $\L$-hard under $\acmred$ reductions.
\end{theorem}
\begin{proof}
    Let $S$ be a finite set of relations.
    If
    $\CSPf{S} \not\in 
    \mAC^0_3$,
    then,
    by \Cref{lem:ac0-cspsat},
    at least 
    one of the following cases hold:
    \begin{enumerate}
        \item $\clonefont{S_{02}} \sseq \Pol(S) \sseq \clonefont{R_2}$ 
            or 
            $\clonefont{S_{12}} \sseq \Pol(S) \sseq \clonefont{R_2}$,
            and $S$ can express the equality relation;
        \item 
            $\clonefont{S_{02}} \not\sseq \Pol(S) \sseq \clonefont{R_2}$
            and
            $\clonefont{S_{12}} \not \sseq \Pol(S) \sseq \clonefont{R_2}$.
        \item 
            $\Pol(S) \not\sseq \clonefont{R_2}$.
    \end{enumerate}
    Since $\CSPf{S}$ is not trivial,
    we obtain that $\CSPf{S}$ is $\L$-hard 
    in the first two cases
    by
    \Cref{lem:equality-lhard,lem:everything-lhard},
    and it's easy to check that $\CSPf{S}$ is also
    $\L$-hard in the third case
    by 
    inspecting
    Post's 
    lattice~(\Cref{fig:post_lattice})
    and the classification theorem~(\Cref{thm:refined-completeness}).
    Since $\L \not\sseq \AC^0$,
    this also implies that, if $\CSPf{S} \in \AC^0$,
    then
    $\clonefont{S_{02}} \sseq \Pol(S) \sseq \clonefont{R_2}$
    or
    $\clonefont{S_{12}} \sseq \Pol(S) \sseq \clonefont{R_2}$,
    and $S$ cannot express the equality relation.
    \Cref{lem:ac0-cspsat}
    again gives
    $\CSPf{S} \in \mAC^0_3$.
\end{proof}

\subsection{Consequences for monotone circuit lower bounds via lifting}
\label{sec:csp-conseq}

We now prove a stronger form of \Cref{thm:cspsat-intro}.
In the previous section, we showed that
$\CSP \cap \AC^0 \sseq \mAC^0$.
In particular, this means that there does not exist a finite set of relations $S$
such that
$\CSPf{S}$ separates $\AC^0$ and $\mNC^1$,
a separation which we proved in \Cref{thm:ac0-fml-graph-intro}.
We will also observe that, if $\CSPf{S} \notin \mNC^2$, then $\CSPf{S}$ is
$\pL$-hard.

\begin{theorem}
    \label{thm:cspsat-hard}
    Let $S$ be a finite set of Boolean relations.
    \begin{enumerate}
        \item 
            If
            $\CSPf{S} \notin 
            \mAC^0_3$
            then $\CSPf{S}$ is $\L$-hard
            under
            $\acmred$ reductions.
        \item 
            If
            $\CSPf{S} \notin \mNC^2$,
            then $\CSPf{S}$ is $\pL$-hard
            under
            $\acmred$ reductions.
    \end{enumerate}
\end{theorem}
\begin{proof}
    Item (1) follows 
            from \Cref{thm:ac0-csp}.
    To prove item (2),
            suppose that $\mdepth{\CSPf{S}} = {\omega(\log^2 n)}$.
            Then, by \Cref{thm:monfml-dichotomy},
            we conclude that 
            $\Pol(S) \sseq \clonefont{L_3}$, or
            $\Pol(S) \sseq \clonefont{V_2}$, or
            $\Pol(S) \sseq \clonefont{E_2}$.
            \Cref{thm:refined-completeness} 
            implies that
            $\CSPf{S}$ is $\pL$-hard.
            \qedhere
\end{proof}

\noindent \textbf{Further Discussion.} We recall the discussion of \Cref{sec:intro_mon_csp}.
We introduced and defined the functions $\CSPf{S}$ in that section,
as a way to capture monotone circuit lower bounds proved via lifting.
This in particular captures the monotone function $\txorsat$,
which was proved in~\cite{gkrs19} to require monotone circuit lower bounds
of size $2^{n^{\Omega(1)}}$ to compute, even though 
$\pL$-machines running in 
polynomial-time
can compute it.
\Cref{thm:cspsat-hard}
proves that this separation between monotone and non-monotone circuit lower
bounds cannot be improved by varying the set of relations $S$, as we argue
below.

There are two ways one could try to find a function in $\AC^0$
with large monotone complexity using a $\CSPSAT$ function.
First, one could try to define a set of relations $S$
such that $\CSPf{S} \in \AC^0$, but the monotone complexity of
$\CSPf{S}$ is large.
However, 
Item (1) of \Cref{thm:cspsat-hard}
proves that this is impossible,
as any $\CSPSAT$ function outside of $\mAC^0$
is $\L$-hard under simple reductions and, therefore, cannot be computed in $\AC^0$.

Secondly, one could try to be apply the arguments of
\Cref{s:constant-depth},
consisting of a padding trick and a simulation
theorem.
When $S$ is the set of 3XOR relations,
then indeed we obtain a function in $\AC^0[\xor]$ with superpolynomial
monotone circuit complexity,
as proved in \Cref{thm:xor-sat}.
However, 
Item (2) of \Cref{thm:cspsat-hard}
proves that this is best possible,
as any $\CSPSAT$ function which admits a superpolynomial monotone circuit
lower bound must be $\pL$-hard and, therefore,
at least as hard as $\txorsat$ for non-monotone circuits.
Item (2) also shows that even $\CSPSAT$ functions with a $\omega(\log^2 n)$
monotone depth lower bound must be $\pL$-hard,
which suggests that the arguments of
\Cref{s:constant-depth} applied to a $\CSPSAT$ function
are not able to prove the separation of \Cref{thm:ac0-fml-graph}.

A caveat to these impossibility results is in order.
First, 
we only study Boolean-valued CSPs here,
though the framework of lifting can also be applied in the context of
non-Boolean CSPs.
It's not clear if non-Boolean CSPs exhibit the same dichotomies
for monotone computation
we proved in this section.
We remark that Schaefer's dichotomy for Boolean-valued CSPs~\cite{DBLP:conf/stoc/Schaefer78}
has been extended to non-Boolean
CSPs~\cite{DBLP:conf/focs/Zhuk17,DBLP:conf/focs/Bulatov17}.

Secondly,
the instances of $\CSPSAT$ generated by lifting
do not cover the entirety of the minterms and maxterms of $\CSPSAT$.
In particular, our results do not rule out the possibility
that a clever interpolation of the instances generated by lifting
may give rise to a function that is easier to compute by non-monotone
circuits,
and therefore bypasses the hardness results of~\Cref{thm:cspsat-hard}.
One example is the Tardós function~\cite{tardos_1988}.
A lifting theorem applied to a Pigeonhole Principle formula can be used to prove 
a lower bound on the size of monotone circuits that accept cliques of size
$k$ and reject graphs that are $(k-1)$-colorable, for some $k = n^\eps$~\cite{DBLP:journals/combinatorica/RazM99,susanna_2023}.
A natural interpolation for these instances would be the $\kclq$
function, which, being $\NP$-complete, would be related to an $\NP$-complete
$\CSPSAT$.
However, as proved by~\cite{tardos_1988}, there is a monotone function
in
$\P$ which has the same output behaviour over these instances.

\bibliographystyle{alpha}	
\bibliography{refs}

\normalsize

\appendix

\section{A Lower Bound for $\txorsat$ Using the Approximation Method}
\label{sec:xorsat-appr}

As discussed in \Cref{sec:intro_mon_csp},
\citep{gkrs19}
obtained an exponential lower bound on the monotone
circuit size of the function $\txorsat$ using techniques from communication
complexity and lifting. Here we observe that a weaker but still
super-polynomial lower bound  can be proved using the approximation method.

First, we recall the
function $\OddFactor_n : \blt^{\binom{n}{2}} \to \blt$
of \Cref{s:padding_graph},
which accepts 
a 
given graph
if the graph contains an \emph{odd factor},
which is a spanning subgraph in which the degree of every
vertex is odd.
For convenience,
in this section we consider a weaker version of $\OddFactor$, which takes as an input a
\emph{bipartite} graph with $n$ vertices on each part,
and accepts if the graph contains an odd factor.
Let $\BipOddFactor_n : \blt^{n^2} \to \blt$
be this function.
We remark that the lower bounds of
Babai, Gál and Wigderson~\cite{DBLP:journals/combinatorica/BabaiGW99}
for $\OddFactor$ (\Cref{thm:bip-oddfactor-lb})
also hold for $\BipOddFactor$.
The proof of the monotone circuit lower bound in particular
is essentially Razborov's lower bound for $\Matching$ via the approximation method~\cite{razborov1985lower}.

\begin{theorem}[\cite{DBLP:journals/combinatorica/BabaiGW99}]
    \label{thm:oddfactor}
    We have
    $$\Cmon{\BipOddFactor_n} = n^{\Omega(\log n)}
    \quad \text{and} \quad 
    \mdepth{\BipOddFactor_n} = \Omega(n).$$
\end{theorem}

We can reduce 
$\BipOddFactor$
to $\txorsat$ by
noting that computing $\BipOddFactor_n(M)$ on a given matrix
$M \in \blt^{n^2}$ is computationally equivalent to deciding the satisfiability of the following $\mathbb{F}_2$ linear system over variables $\{x_{ij}\}$:
\begin{itemize}
    \item For all $i \in [n]$:
        $\bigxor_{k=1}^n x_{ik} = 1$;
    \item For all $j \in [n]$:
        $\bigxor_{k=1}^n x_{kj} = 1$;
    \item For all $i,j \in [n]$ such that $M_{ij} = 0$:
        $x_{ij} = 0$.
\end{itemize}

We can then use a circuit for $\txorsat$ to solve this system
by using a standard trick of introducing new variables
to reduce the number of variables that appear in each equation,
as done in the textbook reduction from $\SAT$ to 3-$\SAT$. As the corresponding reductions turn out to be monotone, this implies monotone circuit and formula lower bounds for $\txorsat$.
We note that a somewhat similar   argument (in the non-monotone setting) appears in Feder and
Vardi~\cite[Theorem 30]{DBLP:journals/siamcomp/FederV98} regarding
constraint satisfaction problems with the ability to count.

In order to formalise this argument, we will need the following definition and results.

\begin{definition}
    \label{def:dual}
    Let $f$ be a Boolean function.
    We define $\dual(f) : x \mapsto \neg f(\neg x)$
    as the \emph{dual} of $f$.
\end{definition}

\begin{lemma}
    \label{lem:dual}
    Let $f$ be a monotone Boolean function.
    We have
    $\Cmon{f} = \Cmon{\dual(f)}$
    and
    $\mdepth{f} = \mdepth{\dual(f)}$.
\end{lemma}
\begin{proof}
    The idea is to push negations to the bottom and eliminate double negations at the input layer. In other words, applying De Morgan rules, we can convert any $\set{\land,\lor}$-circuit computing
    $f$ into a circuit computing $\dual(f)$ 
    by swapping $\land$-gates for
    $\lor$-gates, and vice-versa.
    Moreover, this transformation preserves the depth of the circuit.
\end{proof}

We are ready to describe a monotone reduction from the function $\BipOddFactor_n$ 
to $\txorsat$, which implies the desired lower bounds.

\begin{theorem}
    \label{thm:xorsat-lb}
    There exists $\eps > 0$
    such that
    $$\Cmon{\txorsat} = n^{\Omega(\log n)}
    \quad \text{and} \quad 
    \mdepth{\txorsat} = \Omega(n^\eps).$$
\end{theorem}
\begin{proof}
    Recall that the
    value of the function
    $\BipOddFactor_n(M)$ on a given matrix
    $M \in \blt^{n^2}$ is equal to 1 if the following system is
    satisfiable:
    \begin{itemize}
        \item For all $i \in [n]$:
            $\bigxor_{k=1}^n x_{ik} = 1$;
        \item For all $j \in [n]$:
            $\bigxor_{k=1}^n x_{kj} = 1$;
        \item For all $i,j \in [n]$ such that $M_{ij} = 0$:
            $x_{ij} = 0$.
    \end{itemize}
    We introduce some extra variables to reduce the number of variables in
    each equation in the following way.
    For every $i \in [n]$,
    introduce variables 
    $z_{i1},\dots,z_{i(n-1)}$
    and the equations
    \begin{align*}
        z_{i1} 
        &= 
        x_{i1} \xor x_{i2},
        \\
        z_{i2} 
        &= 
        z_{i1} \xor x_{i3},
        \\
        \dots
        \\
        z_{i,(n-1)} 
        &= 
        z_{i,(n-2)} \xor x_{i,n},
        \\
        z_{i,(n-1)}
        &= 
        1.
    \end{align*}
    Now note that these equations imply
    $z_{i,(n-2)} = \bigxor_{k=1}^n x_{ik} = 1$.
    For each ``column'' equation
    $\bigxor_{k=1}^n x_{kj} = 1$,
    we also add variables
    $w_{j1},\dots, w_{j(n-1)}$ as above.
    In total, we add
    at most
    $2n^2$ variables
    and
    $2n^2$ equations.
    Therefore, there is a system of linear equations
    on $O(n^2)$ variables, where each constraint contains at most $3$ variables, which is satisfiable
    if and only if
    $\BipOddFactor_n(M)=1$.
    Moreover, it is easy to see that the characteristic vector $\alpha$
    of the set of 
    equations of this system
    can be computed from $M$
    by an \emph{anti-monotone} projection, as we activate a constraint that depends on the input when $M_{ij} = 0$.

    Now let 
    $f = \dual(\txorsat)$ and $\beta = \neg \alpha$. Since, by definition, $\txorsat$ accepts \emph{unsatisfiable} systems,
    we get $\BipOddFactor_n(M) = \neg \txorsat(\alpha) = f(\beta)$ and
    that $\beta$ is a \emph{monotone} projection of $M$.
    Therefore, by \cref{lem:dual},
    we obtain
    \begin{equation*}
        \Cmon{\BipOddFactor_n} \leq \Cmon{\txorsat}
    \end{equation*}
    and
    \begin{equation*}
        \mdepth{\BipOddFactor_n} \leq \mdepth{\txorsat}.
        \qedhere
    \end{equation*}
\end{proof}

\section{Schaefer's Theorem in Monotone Complexity}
\label{sec:schaefer_proof}

\subsection{Connectivity and generation functions}

We recall the definitions of two prominent monotone Boolean functions
that have efficient monotone circuits.
Let $\stconn : \blt^{n^2} \to \blt$
be the function that outputs $1$ on a given directed graph $G$
if there exists a path from $1$ to $n$ in $G$.
Let $\GEN : \blt^{n^3} \to \blt$
be the Boolean function which receives a set $T$ of triples 
$(i,j,k) \in [n^3]$,
and outputs $1$ if 
$n \in S$, where $S \sseq [n]$ is the 
set generated with the following rules:

\begin{itemize}
    \item \emph{Axiom:} $1 \in S$,
    \item \emph{Generation:} If $i,j \in S$ and $(i,j,k) \in T$, then
        $k \in S$.
\end{itemize}

The following upper bounds are well-known and easy to prove.

\begin{theorem}
    [\protect{\cite[Exercise 7.3]{jukna_2012}, \cite{DBLP:journals/combinatorica/RazM99}}]
    \label{thm:stconn-gen}
    We have $\stconn \in \mNL$ and
    $\GEN \in \mSIZE[\poly]$.
\end{theorem}

\subsection{Proof of reduction lemmas}
\label{sec:reduction_monotone}

Here we present monotonised versions of the proofs of
\cite[Propositions 2.2 - 2.4]{playing_two},
which give a simplified presentation of the results
of~\cite{DBLP:conf/stoc/Schaefer78}.

\inclcoq*

\begin{proof}
    If $\COQ(S_1) \sseq \COQ(S_2)$, then
    each relation of $S_1$ can be represented as a conjunctive query
    over $S_2$.
    Let $F_1$ be a $S_1$-formula.
    For each constraint $C_1$ of $F_1$, there exists
    a formula $\varphi(C_1)$ in $\CNF(S_2)$
    such that
    $C_1$ is a projection of $\varphi(C_1)
    $ (i.e., $C_1$ is a conjunctive query of
    $\varphi(C_1)$).
    However, note that $C_1$ is satisfiable if and only if
    $\varphi(C_1)$ is satisfiable.
    So we can replace the constraint $C_1$ by the \emph{set} of constraints
    in $\varphi(C_1)$.
    Doing this for every constraint in $F_1$, we obtain an $S_2$-formula $F_2$
    which is satisfiable iff $F_1$ is satisfiable.

    Now note that, to decide if a given constraint application $C$ of
    $S_2$
    is in the reduction,
    it suffices
    to check if there exists a $S_1$-constraint $C_1$
    in $F_1$
    such that
    $C$ is in $\varphi(C_1)$.
    Using non-uniformity, this can be easily done
    by an OR over the relevant input bits.

    Finally, we observe that, since the arities of each relation in
    $S_1$ and $S_2$ are constant, we only add a constant number of
    variables for each constraint
    to represent $S_1$-formulas with conjunctive queries over $S_2$-formulas.
\end{proof}

\begin{lemma}
    \label{prop:eql}
    Let $S$ be a set of Boolean relations.
    We have
    $\CSPf{S \cup \set{=}} \mnlmred \CSPf{S}$.
\end{lemma}
\begin{proof}
    Let $F$ be a $(S \cup \set{=})$-formula on $n$ variables
    given as an input.
    Remember that $F$ is given as a Boolean vector $\alpha$,
    where each bit of $\alpha$
    represents
    the presence of a constraint application on $n$ variables
    from $S \cup \set{=}$.
    We first build an undirected graph $G$ with the variables
    $x_1,\dots,x_n$ as vertices,
    and we put an edge between $x_i$ and $x_j$
    if the constraint $x_i = x_j$ appears in $F$.
    Note that $G$ can be constructed by a monotone projection from
    $F$.

    Let $R \in S$ and let
    $C = R(x_1,\dots,x_n)$ be a constraint application of $R$.
    If $C$ appears in $F$,
    we add to the system every constraint of the form
    $C' = R(y_1,\dots,y_n)$
    such that,
    for every $i \in [n]$,
    there exists a path from $x_i$ to $y_i$ in the graph $G$.
    In this case, we say that $C$ \emph{generates} $C'$.
    Let $F_2$ be the formula that contains all
    non-equality constraints of $F$,
    and all the non-equality constraints
    generated by a constraint in $F$.
    It's not hard to see that 
    $F$ is satisfiable if and only if $F_2$ is satisfiable,
    and therefore the reduction is correct.

    Moreover, the reduction can be done in monotone $\NL$
    using the monotone $\NL$ algorithm for $\stconn$
    (\Cref{thm:stconn-gen}).
    Indeed, there are at most $n^k$ constraint applications of a given
    relation $R$ of arity $k$. Therefore,
    to decide if a constraint $C' = R(y_1,\dots,y_n)$ appears in $F_2$,
    it suffices to check if there exists a constraint application 
    of $R$ in $F$ which generates $C'$.
    This can be checked with $n^k$ calls to $\stconn$.
\end{proof}

\polym*
\begin{proof}
    If $\Pol(S_2) \sseq \Pol(S_1)$, then 
    from \Cref{lem:galois}
    (\Cref{item:inv-invert,item:coq-sseq,item:invpol})
    we obtain
    $\COQ(S_1) \sseq \ccln{S_1} \sseq \ccln{S_2} = \COQ(S_2 \cup \set{=})$.
    Therefore, by~\Cref{prop:incl-coq,prop:eql} 
    we can do the following chain of reductions in monotone $\NL$:
    \begin{equation*}
        \CSPf{S_1}
        \mormred
        \CSPf{S_2 \cup \set{=}}
        \mnlmred
        \CSPf{S_2}.
        \qedhere
    \end{equation*}
\end{proof}

\subsection{Monotone circuit upper bounds}
\label{sec:monotone_upper}

We restate and prove the theorem.

\monschaeffer*

\begin{proof}
    We prove each case separately.

    \emph{Proof of (1)}.
    We first observe that
    $\thornsat$ 
    (see definition in \Cref{sec:csp-dichotomy}, \emph{Dichotomy for formulas})
    can be solved by a reduction to
    $\GEN \in \mSIZE[\poly]$.
    Indeed, we interpret each constraint of the form
    $(\neg x_i \lor \neg x_j \lor x_k)$
    (which is equivalent to 
    $x_i \land x_j \implies x_k$)
    as a triple
    $(i,j,k)$,
    and constraints of the form
    $x_i$
    as a triple
    $(0,0,i)$.
    Let $S \sseq \set{0,1,2\dots,n}$ be the set generated
    by these triples, applying the generation rules of $\GEN$,
    using $0 \in S$ as the axiom.
    It suffices to check that there exists
    some constraint of the form
    $\neg x_i \lor \neg x_j \lor \neg x_k$,
    such that $\set{i,j,k} \sseq S$.
    This process can be done with polynomial-size
    monotone circuits, invoking $\GEN$.
    Therefore, it follows
    from \Cref{thm:stconn-gen}
    that
    $\thornsat \in \mSIZE[\poly]$.

    Moreover, we recall that, if $\clonefont{E_2} \sseq \Pol(S)$, then
    $S \sseq \COQ(\hornt)$
    (in other words, every $S$-formula can be written as a set of 3-Horn
    equations) -- see, e.g,~\cite[Lemma 4.8]{DBLP:books/daglib/0004131}.
    Therefore, from 
    \Cref{item:coq-idempotent,item:coq-sseq} of \Cref{lem:galois}
    and
    \Cref{prop:incl-coq},
    we conclude that
    $\CSPf{S} \mormred \thornsat \in \mSIZE[\poly]$.

    Now recall that, if $\clonefont{V_2} \sseq \Pol(S)$,
    then
    $S \sseq \COQ(\ahornt)$,
    where
    $\ahornt$ is the set of width-3 Anti-Horn relations
    (i.e., $\ahornt = 
    \set{
        (x_1 \lor x_2 \lor \neg x_3),
        (x_1 \lor x_2 \lor x_3),
        (\neg x)
    }$; see~\cite[Lemma 4.8]{DBLP:books/daglib/0004131} for a proof of this
    observation).
    But note that an
    $\ahornt$-formula $\varphi$
    is satisfiable if and only if
    the $\hornt$-formula 
    $\varphi(\neg x_1,\dots,\neg x_n)$
    is satisfiable.
    Therefore by~\Cref{prop:incl-coq}
    and
    \Cref{item:coq-idempotent,item:coq-sseq} of \Cref{lem:galois},
    we have
    $\CSPf{S} \mormred 
    \CSPf{\ahornt}
    \mprojred \thornsat \in \mSIZE[\poly]$.

    \emph{Proof of (2)}.
    We first prove the case $\clonefont{D_2} \sseq \Pol(S)$.
    Let $\dSAT = \CSPf{\Gamma}$,
    where $\Gamma = \set{(x_1 \lor x_2), (x_1 \lor \neg x_2), (\neg x_1 \lor
    \neg x_2)}$.
    It's easy to check that the standard reduction from
    $\dSAT$ to $\stconn$ can be done in monotone $\NL$
    (see~\cite[Theorem 4]{DBLP:journals/mst/JonesLL76}).
    Therefore, it follows from \Cref{thm:stconn-gen}
    that
    $\dSAT \in \mNL$.
    Now, recall that, if
    $\clonefont{D_2} \sseq \Pol(S)$, then
    $S \sseq \COQ(\Gamma)$
    (see, e.g., \cite[Lemma 4.9]{DBLP:books/daglib/0004131}).
    Therefore,
    from \Cref{prop:incl-coq} and
    \Cref{item:coq-idempotent,item:coq-sseq} of \Cref{lem:galois},
    we conclude
    $\CSPf{S} \in \mNL$.

    We now suppose that $\clonefont{S_{00}} \sseq \Pol(S)$.
    We check that the proof of~\cite[Lemma 3.4]{DBLP:journals/jcss/AllenderBISV09}
    gives a monotone circuit.
    If $\clonefont{S_{00}} \sseq \Pol(S)$, then
    there exists $k \geq 2$ such that
    $\clonefont{S_{00}}^k \sseq \Pol(S)$
    (that's because there does not exist a finite set of relations
    $S$ such that $\Pol(S) = \clonefont{S_{00}}$).
    Note that
    $\clonefont{S_{00}}^k = \Pol(\Gamma)$,
    where $\Gamma = \set{\OR^k, x, \neg x, \to, =}$.
    We show below how to decide if a $\Gamma$-formula is unsatisfiable 
    in monotone $\NL$.
    The result then follows from~\Cref{prop:poly}.

    Let $F$ be a given $\Gamma$-formula with $n$ variables.
    We first construct a directed graph $G$,
    with vertex set $\set{x_1,\dots,x_n}$,
    and with arcs 
    $(x_i,x_j)$
    if $x_i \to x_j$ is a constraint of $F$,
    and 
    arcs
    $(x_i,x_j)$
    and
    $(x_j,x_i)$
    if
    $x_i = x_j$ is a constraint of $F$.
    This can be done with a monotone projection.
    Observe that a $\Gamma$-formula $F$
    is unsatisfiable if, and only if,
    there exists a constraint of the form $x_{i_1} \lor \dots \lor x_{i_k}$ in $F$,
    such that there exists a path from some $x_{i_j}$ to a constraint
    $\neg y$ in $F$.
    This can be checked in monotone $\NL$ by \Cref{thm:stconn-gen}.

    The case $\clonefont{S_{10}} \sseq \Pol(S)$ is analogous.
\end{proof}

\section{Background on Post's Lattice and Clones}
\label{sec:clone-background}

In this section, we include
 the definitions of the various clones that are used in the
paper,
as well as
a figure of Post's lattice, which can be helpful when checking the proofs
of
\Cref{sec:CSPs}.

Let
$\rightarrow : (x,y) \mapsto (\neg x \lor y)$.
Let also
$\leftrightarrow : (x,y) \mapsto \neg (x \xor y)$
and
$\clonefont{id} : x \mapsto x$.
Let $f : \blt^k \to \blt$ be a Boolean function.
We say that 
$f$
is \emph{linear}
if there exists 
$c \in \blt^k$ and $b \in \blt$
such that
$f(x) = \inner{c}{x} + b \pmod 2$.
We say that $f$ is self-dual if $f = \dual(f)$. Let $a \in \blt$.
We say that $f$ is 
\emph{a-reproducing}
if
$f(a,\dots,a) = a$.
We say that
a set $T \sseq \blt^k$ is 
\emph{$a$-separating}
if there exists $i \in [k]$
such that
$x_i = a$ for all $x \in T$.
We say that
$f$ is 
\emph{$a$-separating}
if
$f^{-1}(a)$ is $a$-separating.
We say that $f$ is 
\emph{a-separating of degree $k$}
if every $T \sseq f^{-1}(a)$
such that $\card{T} = k$ is $a$-separating. The \emph{basis} of a clone $B$ is a set of Boolean functions $F$
such that $B = [F]$.

\begin{figure}[hbtp]
    \centering
    \begin{small}
        \begin{longtable}[]{p{0.1\textwidth}p{0.5\textwidth}p{0.3\textwidth}}
            \toprule
            \begin{minipage}[b]{0.12\columnwidth}\raggedright
                Name\strut
            \end{minipage} & \begin{minipage}[b]{0.41\columnwidth}\raggedright
                Definition\strut
            \end{minipage} & \begin{minipage}[b]{0.38\columnwidth}\raggedright
                Base\strut
            \end{minipage}\tabularnewline
            \midrule
            \endhead
            \begin{minipage}[t]{0.12\columnwidth}\raggedright
                \(\mathrm{BF}\)\strut
            \end{minipage} & \begin{minipage}[t]{0.41\columnwidth}\raggedright
                All Boolean functions\strut
            \end{minipage} & \begin{minipage}[t]{0.38\columnwidth}\raggedright
                \(\set{\lor,\land,\neg}\)\strut
            \end{minipage}\tabularnewline
            \begin{minipage}[t]{0.12\columnwidth}\raggedright
                \(\clonefont{R_0}\)\strut
            \end{minipage} & \begin{minipage}[t]{0.41\columnwidth}\raggedright
                \(\set{f \in \clonefont{BF} : f \text{ is $0$-reproducing}}\)\strut
            \end{minipage} & \begin{minipage}[t]{0.38\columnwidth}\raggedright
                \(\set{\land,\xor}\)\strut
            \end{minipage}\tabularnewline
            \begin{minipage}[t]{0.12\columnwidth}\raggedright
                \(\clonefont{R_1}\)\strut
            \end{minipage} & \begin{minipage}[t]{0.41\columnwidth}\raggedright
                \(\set{f \in \clonefont{BF} : f \text{ is $1$-reproducing}}\)\strut
            \end{minipage} & \begin{minipage}[t]{0.38\columnwidth}\raggedright
                \(\set{\lor,\leftrightarrow}\)\strut
            \end{minipage}\tabularnewline
            \begin{minipage}[t]{0.12\columnwidth}\raggedright
                \(\clonefont{R_2}\)\strut
            \end{minipage} & \begin{minipage}[t]{0.41\columnwidth}\raggedright
                \(\clonefont{R_1 \cap R_0}\)\strut
            \end{minipage} & \begin{minipage}[t]{0.38\columnwidth}\raggedright
                \(\set{\lor, x \land (y \leftrightarrow z)}\)\strut
            \end{minipage}\tabularnewline
            \begin{minipage}[t]{0.12\columnwidth}\raggedright
                \(\clonefont{M}\)\strut
            \end{minipage} & \begin{minipage}[t]{0.41\columnwidth}\raggedright
                \(\set{f \in \clonefont{BF} : f \text{ is monotonic}}\)\strut
            \end{minipage} & \begin{minipage}[t]{0.38\columnwidth}\raggedright
                \(\set{\lor,\land,0,1}\)\strut
            \end{minipage}\tabularnewline
            \begin{minipage}[t]{0.12\columnwidth}\raggedright
                \(\clonefont{M_1}\)\strut
            \end{minipage} & \begin{minipage}[t]{0.41\columnwidth}\raggedright
                \(\clonefont{M} \cap \clonefont{R_1}\)\strut
            \end{minipage} & \begin{minipage}[t]{0.38\columnwidth}\raggedright
                \(\set{\lor,\land,1}\)\strut
            \end{minipage}\tabularnewline
            \begin{minipage}[t]{0.12\columnwidth}\raggedright
                \(\clonefont{M_0}\)\strut
            \end{minipage} & \begin{minipage}[t]{0.41\columnwidth}\raggedright
                \(\clonefont{M} \cap \clonefont{R_0}\)\strut
            \end{minipage} & \begin{minipage}[t]{0.38\columnwidth}\raggedright
                \(\set{\lor,\land,0}\)\strut
            \end{minipage}\tabularnewline
            \begin{minipage}[t]{0.12\columnwidth}\raggedright
                \(\clonefont{M_2}\)\strut
            \end{minipage} & \begin{minipage}[t]{0.41\columnwidth}\raggedright
                \(\clonefont{M} \cap \clonefont{R_2}\)\strut
            \end{minipage} & \begin{minipage}[t]{0.38\columnwidth}\raggedright
                \(\set{\lor,\land}\)\strut
            \end{minipage}\tabularnewline
            \begin{minipage}[t]{0.12\columnwidth}\raggedright
                \(\clonefont{S_0^n}\)\strut
            \end{minipage} & \begin{minipage}[t]{0.41\columnwidth}\raggedright
                \(\set{f \in \clonefont{BF} : f \text{ is $0$-separating of degree $n$}}\)\strut
            \end{minipage} & \begin{minipage}[t]{0.38\columnwidth}\raggedright
                \(\set{\rightarrow,\dual(h_n)}\)\strut
            \end{minipage}\tabularnewline
            \begin{minipage}[t]{0.12\columnwidth}\raggedright
                \(\clonefont{S_0}\)\strut
            \end{minipage} & \begin{minipage}[t]{0.41\columnwidth}\raggedright
                \(\set{f \in \clonefont{BF} : f \text{ is $0$-separating}}\)\strut
            \end{minipage} & \begin{minipage}[t]{0.38\columnwidth}\raggedright
                \(\set{\rightarrow}\)\strut
            \end{minipage}\tabularnewline
            \begin{minipage}[t]{0.12\columnwidth}\raggedright
                \(\clonefont{S_1^n}\)\strut
            \end{minipage} & \begin{minipage}[t]{0.41\columnwidth}\raggedright
                \(\set{f \in \clonefont{BF} : f \text{ is $1$-separating of degree $n$}}\)\strut
            \end{minipage} & \begin{minipage}[t]{0.38\columnwidth}\raggedright
                \(\set{x \land \ovl{y} ,h_n}\)\strut
            \end{minipage}\tabularnewline
            \begin{minipage}[t]{0.12\columnwidth}\raggedright
                \(\clonefont{S_1}\)\strut
            \end{minipage} & \begin{minipage}[t]{0.41\columnwidth}\raggedright
                \(\set{f \in \clonefont{BF} : f \text{ is $1$-separating}}\)\strut
            \end{minipage} & \begin{minipage}[t]{0.38\columnwidth}\raggedright
                \(\set{x \land \ovl{y}}\)\strut
            \end{minipage}\tabularnewline
            \begin{minipage}[t]{0.12\columnwidth}\raggedright
                \(\clonefont{S_{02}^n}\)\strut
            \end{minipage} & \begin{minipage}[t]{0.41\columnwidth}\raggedright
                \(\clonefont{S_0^n} \cap \clonefont{R_2}\)\strut
            \end{minipage} & \begin{minipage}[t]{0.38\columnwidth}\raggedright
                \(\set{x \lor (y \land \ovl{z}), \dual(h_n)}\)\strut
            \end{minipage}\tabularnewline
            \begin{minipage}[t]{0.12\columnwidth}\raggedright
                \(\clonefont{S_{02}}\)\strut
            \end{minipage} & \begin{minipage}[t]{0.41\columnwidth}\raggedright
                \(\clonefont{S_0} \cap \clonefont{R_2}\)\strut
            \end{minipage} & \begin{minipage}[t]{0.38\columnwidth}\raggedright
                \(\set{x \lor (y \land \ovl{z})}\)\strut
            \end{minipage}\tabularnewline
            \begin{minipage}[t]{0.12\columnwidth}\raggedright
                \(\clonefont{S_{01}^n}\)\strut
            \end{minipage} & \begin{minipage}[t]{0.41\columnwidth}\raggedright
                \(\clonefont{S_0^n} \cap \clonefont{M}\)\strut
            \end{minipage} & \begin{minipage}[t]{0.38\columnwidth}\raggedright
                \(\set{\dual(h_n), 1}\)\strut
            \end{minipage}\tabularnewline
            \begin{minipage}[t]{0.12\columnwidth}\raggedright
                \(\clonefont{S_{01}}\)\strut
            \end{minipage} & \begin{minipage}[t]{0.41\columnwidth}\raggedright
                \(\clonefont{S_0} \cap \clonefont{M}\)\strut
            \end{minipage} & \begin{minipage}[t]{0.38\columnwidth}\raggedright
                \(\set{x \lor (y \land z), 1}\)\strut
            \end{minipage}\tabularnewline
            \begin{minipage}[t]{0.12\columnwidth}\raggedright
                \(\clonefont{S_{00}^n}\)\strut
            \end{minipage} & \begin{minipage}[t]{0.41\columnwidth}\raggedright
                \(\clonefont{S_0^n} \cap \clonefont{R_2} \cap \clonefont{M}\)\strut
            \end{minipage} & \begin{minipage}[t]{0.38\columnwidth}\raggedright
                \(\set{x \lor (y \land z), \dual(h_n)}\)\strut
            \end{minipage}\tabularnewline
            \begin{minipage}[t]{0.12\columnwidth}\raggedright
                \(\clonefont{S_{00}}\)\strut
            \end{minipage} & \begin{minipage}[t]{0.41\columnwidth}\raggedright
                \(\clonefont{S_0} \cap \clonefont{R_2} \cap \clonefont{M}\)\strut
            \end{minipage} & \begin{minipage}[t]{0.38\columnwidth}\raggedright
                \(\set{x \lor (y \land z)}\)\strut
            \end{minipage}\tabularnewline
            \begin{minipage}[t]{0.12\columnwidth}\raggedright
                \(\clonefont{S_{12}^n}\)\strut
            \end{minipage} & \begin{minipage}[t]{0.41\columnwidth}\raggedright
                \(\clonefont{S_1^n} \cap \clonefont{R_2}\)\strut
            \end{minipage} & \begin{minipage}[t]{0.38\columnwidth}\raggedright
                \(\set{x \land (y \lor \ovl{z}), h_n}\)\strut
            \end{minipage}\tabularnewline
            \begin{minipage}[t]{0.12\columnwidth}\raggedright
                \(\clonefont{S_{12}}\)\strut
            \end{minipage} & \begin{minipage}[t]{0.41\columnwidth}\raggedright
                \(\clonefont{S_1} \cap \clonefont{R_2}\)\strut
            \end{minipage} & \begin{minipage}[t]{0.38\columnwidth}\raggedright
                \(\set{x \land (y \lor \ovl{z})}\)\strut
            \end{minipage}\tabularnewline
            \begin{minipage}[t]{0.12\columnwidth}\raggedright
                \(\clonefont{S_{11}^n}\)\strut
            \end{minipage} & \begin{minipage}[t]{0.41\columnwidth}\raggedright
                \(\clonefont{S_1^n} \cap \clonefont{M}\)\strut
            \end{minipage} & \begin{minipage}[t]{0.38\columnwidth}\raggedright
                \(\set{h_n, 0}\)\strut
            \end{minipage}\tabularnewline
            \begin{minipage}[t]{0.12\columnwidth}\raggedright
                \(\clonefont{S_{11}}\)\strut
            \end{minipage} & \begin{minipage}[t]{0.41\columnwidth}\raggedright
                \(\clonefont{S_1} \cap \clonefont{M}\)\strut
            \end{minipage} & \begin{minipage}[t]{0.38\columnwidth}\raggedright
                \(\set{x \land (y \lor z), 0}\)\strut
            \end{minipage}\tabularnewline
            \begin{minipage}[t]{0.12\columnwidth}\raggedright
                \(\clonefont{S_{10}^n}\)\strut
            \end{minipage} & \begin{minipage}[t]{0.41\columnwidth}\raggedright
                \(\clonefont{S_1^n} \cap \clonefont{R_2} \cap \clonefont{M}\)\strut
            \end{minipage} & \begin{minipage}[t]{0.38\columnwidth}\raggedright
                \(\set{x \land (y \lor z), h_n}\)\strut
            \end{minipage}\tabularnewline
            \begin{minipage}[t]{0.12\columnwidth}\raggedright
                \(\clonefont{S_{10}}\)\strut
            \end{minipage} & \begin{minipage}[t]{0.41\columnwidth}\raggedright
                \(\clonefont{S_1} \cap \clonefont{R_2} \cap \clonefont{M}\)\strut
            \end{minipage} & \begin{minipage}[t]{0.38\columnwidth}\raggedright
                \(\set{x \land (y \lor z)}\)\strut
            \end{minipage}\tabularnewline
            \begin{minipage}[t]{0.12\columnwidth}\raggedright
                \(\clonefont{D}\)\strut
            \end{minipage} & \begin{minipage}[t]{0.41\columnwidth}\raggedright
                \(\set{f \in \clonefont{BF} : f \text{ is self-dual}}\)\strut
            \end{minipage} & \begin{minipage}[t]{0.38\columnwidth}\raggedright
                \(\set{(x \land \ovl{y}) \lor (x \land \ovl{z}) \lor (\ovl{y} \land \ovl{z})}\)\strut
            \end{minipage}\tabularnewline
            \begin{minipage}[t]{0.12\columnwidth}\raggedright
                \(\clonefont{D_1}\)\strut
            \end{minipage} & \begin{minipage}[t]{0.41\columnwidth}\raggedright
                \(\clonefont{D} \cap \clonefont{R_2}\)\strut
            \end{minipage} & \begin{minipage}[t]{0.38\columnwidth}\raggedright
                \(\set{(x \land y) \lor (x \land \ovl{z}) \lor ({y} \land \ovl{z})}\)\strut
            \end{minipage}\tabularnewline
            \begin{minipage}[t]{0.12\columnwidth}\raggedright
                \(\clonefont{D_2}\)\strut
            \end{minipage} & \begin{minipage}[t]{0.41\columnwidth}\raggedright
                \(\clonefont{D} \cap \clonefont{M}\)\strut
            \end{minipage} & \begin{minipage}[t]{0.38\columnwidth}\raggedright
                \(\set{(x \land y) \lor (y \land {z}) \lor (x \land {z})}\)\strut
            \end{minipage}\tabularnewline
            \begin{minipage}[t]{0.12\columnwidth}\raggedright
                \(\clonefont{L}\)\strut
            \end{minipage} & \begin{minipage}[t]{0.41\columnwidth}\raggedright
                \(\set{f \in \clonefont{BF} : f \text{ is linear}}\)\strut
            \end{minipage} & \begin{minipage}[t]{0.38\columnwidth}\raggedright
                \(\set{\xor, 1}\)\strut
            \end{minipage}\tabularnewline
            \begin{minipage}[t]{0.12\columnwidth}\raggedright
                \(\clonefont{L_0}\)\strut
            \end{minipage} & \begin{minipage}[t]{0.41\columnwidth}\raggedright
                \(\clonefont{L} \cap \clonefont{R_0}\)\strut
            \end{minipage} & \begin{minipage}[t]{0.38\columnwidth}\raggedright
                \(\set{\xor}\)\strut
            \end{minipage}\tabularnewline
            \begin{minipage}[t]{0.12\columnwidth}\raggedright
                \(\clonefont{L_1}\)\strut
            \end{minipage} & \begin{minipage}[t]{0.41\columnwidth}\raggedright
                \(\clonefont{L} \cap \clonefont{R_1}\)\strut
            \end{minipage} & \begin{minipage}[t]{0.38\columnwidth}\raggedright
                \(\set{\leftrightarrow}\)\strut
            \end{minipage}\tabularnewline
            \begin{minipage}[t]{0.12\columnwidth}\raggedright
                \(\clonefont{L_2}\)\strut
            \end{minipage} & \begin{minipage}[t]{0.41\columnwidth}\raggedright
                \(\clonefont{L} \cap \clonefont{R}\)\strut
            \end{minipage} & \begin{minipage}[t]{0.38\columnwidth}\raggedright
                \(\set{x \xor y \xor z}\)\strut
            \end{minipage}\tabularnewline
            \begin{minipage}[t]{0.12\columnwidth}\raggedright
                \(\clonefont{L_3}\)\strut
            \end{minipage} & \begin{minipage}[t]{0.41\columnwidth}\raggedright
                \(\clonefont{L} \cap \clonefont{D}\)\strut
            \end{minipage} & \begin{minipage}[t]{0.38\columnwidth}\raggedright
                \(\set{x \xor y \xor z \xor 1}\)\strut
            \end{minipage}\tabularnewline
            \begin{minipage}[t]{0.12\columnwidth}\raggedright
                \(\clonefont{V}\)\strut
            \end{minipage} & \begin{minipage}[t]{0.41\columnwidth}\raggedright
                \(\set{f \in \clonefont{BF} : f \text{ is constant or an $n$-ary OR function}}\)\strut
            \end{minipage} & \begin{minipage}[t]{0.38\columnwidth}\raggedright
                \(\set{\lor,0,1}\)\strut
            \end{minipage}\tabularnewline
            \begin{minipage}[t]{0.12\columnwidth}\raggedright
                \(\clonefont{V_0}\)\strut
            \end{minipage} & \begin{minipage}[t]{0.41\columnwidth}\raggedright
                \([\set{\lor}] \cup [\set{0}]\)\strut
            \end{minipage} & \begin{minipage}[t]{0.38\columnwidth}\raggedright
                \(\set{\lor,0}\)\strut
            \end{minipage}\tabularnewline
            \begin{minipage}[t]{0.12\columnwidth}\raggedright
                \(\clonefont{V_1}\)\strut
            \end{minipage} & \begin{minipage}[t]{0.41\columnwidth}\raggedright
                \([\set{\lor}] \cup [\set{1}]\)\strut
            \end{minipage} & \begin{minipage}[t]{0.38\columnwidth}\raggedright
                \(\set{\lor,1}\)\strut
            \end{minipage}\tabularnewline
            \begin{minipage}[t]{0.12\columnwidth}\raggedright
                \(\clonefont{V_2}\)\strut
            \end{minipage} & \begin{minipage}[t]{0.41\columnwidth}\raggedright
                \([\set{\lor}]\)\strut
            \end{minipage} & \begin{minipage}[t]{0.38\columnwidth}\raggedright
                \(\set{\lor}\)\strut
            \end{minipage}\tabularnewline
            \begin{minipage}[t]{0.12\columnwidth}\raggedright
                \(\clonefont{E}\)\strut
            \end{minipage} & \begin{minipage}[t]{0.41\columnwidth}\raggedright
                \(\set{f \in \clonefont{BF} : f \text{ is constant or an $n$-ary AND function}}\)\strut
            \end{minipage} & \begin{minipage}[t]{0.38\columnwidth}\raggedright
                \(\set{\land,0,1}\)\strut
            \end{minipage}\tabularnewline
            \begin{minipage}[t]{0.12\columnwidth}\raggedright
                \(\clonefont{E_0}\)\strut
            \end{minipage} & \begin{minipage}[t]{0.41\columnwidth}\raggedright
                \([\set{\land}] \cup [\set{0}]\)\strut
            \end{minipage} & \begin{minipage}[t]{0.38\columnwidth}\raggedright
                \(\set{\land,0}\)\strut
            \end{minipage}\tabularnewline
            \begin{minipage}[t]{0.12\columnwidth}\raggedright
                \(\clonefont{E_1}\)\strut
            \end{minipage} & \begin{minipage}[t]{0.41\columnwidth}\raggedright
                \([\set{\land}] \cup [\set{1}]\)\strut
            \end{minipage} & \begin{minipage}[t]{0.38\columnwidth}\raggedright
                \(\set{\land,1}\)\strut
            \end{minipage}\tabularnewline
            \begin{minipage}[t]{0.12\columnwidth}\raggedright
                \(\clonefont{E_2}\)\strut
            \end{minipage} & \begin{minipage}[t]{0.41\columnwidth}\raggedright
                \([\set{\land}]\)\strut
            \end{minipage} & \begin{minipage}[t]{0.38\columnwidth}\raggedright
                \(\set{\land}\)\strut
            \end{minipage}\tabularnewline
            \begin{minipage}[t]{0.12\columnwidth}\raggedright
                \(\clonefont{N}\)\strut
            \end{minipage} & \begin{minipage}[t]{0.41\columnwidth}\raggedright
                \([\set{\neg}] \cup [\set{0}] \cup [\set{1}]\)\strut
            \end{minipage} & \begin{minipage}[t]{0.38\columnwidth}\raggedright
                \(\set{\neg, 1}\)\strut
            \end{minipage}\tabularnewline
            \begin{minipage}[t]{0.12\columnwidth}\raggedright
                \(\clonefont{N_2}\)\strut
            \end{minipage} & \begin{minipage}[t]{0.41\columnwidth}\raggedright
                \([\set{\neg}]\)\strut
            \end{minipage} & \begin{minipage}[t]{0.38\columnwidth}\raggedright
                \(\set{\neg}\)\strut
            \end{minipage}\tabularnewline
            \begin{minipage}[t]{0.12\columnwidth}\raggedright
                \(\clonefont{I}\)\strut
            \end{minipage} & \begin{minipage}[t]{0.41\columnwidth}\raggedright
                \([\set{\clonefont{id}}] \cup [\set{0}] \cup [\set{1}]\)\strut
            \end{minipage} & \begin{minipage}[t]{0.38\columnwidth}\raggedright
                \(\set{\clonefont{id}, 0, 1}\)\strut
            \end{minipage}\tabularnewline
            \begin{minipage}[t]{0.12\columnwidth}\raggedright
                \(\clonefont{I_0}\)\strut
            \end{minipage} & \begin{minipage}[t]{0.41\columnwidth}\raggedright
                \([\set{\clonefont{id}}] \cup [\set{0}]\)\strut
            \end{minipage} & \begin{minipage}[t]{0.38\columnwidth}\raggedright
                \(\set{\clonefont{id}, 0}\)\strut
            \end{minipage}\tabularnewline
            \begin{minipage}[t]{0.12\columnwidth}\raggedright
                \(\clonefont{I_1}\)\strut
            \end{minipage} & \begin{minipage}[t]{0.41\columnwidth}\raggedright
                \([\set{\clonefont{id}}] \cup [\set{1}]\)\strut
            \end{minipage} & \begin{minipage}[t]{0.38\columnwidth}\raggedright
                \(\set{\clonefont{id}, 1}\)\strut
            \end{minipage}\tabularnewline
            \begin{minipage}[t]{0.12\columnwidth}\raggedright
                \(\clonefont{I_2}\)\strut
            \end{minipage} & \begin{minipage}[t]{0.41\columnwidth}\raggedright
                \([\set{\clonefont{id}}]\)\strut
            \end{minipage} & \begin{minipage}[t]{0.38\columnwidth}\raggedright
                \(\set{\clonefont{id}}\)\strut
            \end{minipage}\tabularnewline
            \bottomrule
        \end{longtable}
    \end{small}
    \caption{
        Table of all closed classes of Boolean functions, and their bases.
        Here, $h_n$ denotes the function
        $h_n(x_1,\dots,x_{n+1}) = \bigvee_{i=1}^{n+1} \bigwedge_{j = 1,
        j \neq i}^{n+1} x_j$.
        See \Cref{def:dual} for the definition of $\dual(\cdot)$.
        The same
        table appears in
    \cite[Table 1]{DBLP:journals/jcss/AllenderBISV09}.}
    \label{fig:clone_table}
\end{figure}

\end{document}